\newcount\Comments  
\documentclass{article}
\usepackage[margin=1in]{geometry}
\usepackage{soul}
\usepackage{dsfont}
\usepackage{url}
\usepackage[utf8]{inputenc}
\usepackage{graphicx}
\usepackage{amsmath}
\usepackage{algorithmic}
\usepackage{amsthm}
\usepackage{amssymb}
\usepackage[ruled,vlined]{algorithm2e}
\usepackage{graphicx}
\usepackage{comment}
\usepackage[usenames,dvipsnames,svgnames,table]{xcolor}
\usepackage[colorinlistoftodos]{todonotes}
\usepackage[colorlinks=true, allcolors=NavyBlue]{hyperref}
\usepackage{bm}
\usepackage{natbib}

\newcommand{\E}[0]{\mathbb{E}}

\newcommand{\ignore}[1]{}
\newcommand{\kibitz}[2]{\ifnum\Comments=1\textcolor{#1}{#2}\fi}

\renewcommand{\hat}[1]{\widehat{#1}}

\def\argmax{\qopname\relax n{argmax}}

\newtheorem{theorem}{Theorem}[section]

\newtheorem{lemma}[theorem]{Lemma}

\newtheorem{definition}[theorem]{Definition}

\newtheorem{example}[theorem]{Example}

\newtheorem*{conjecture*}{Conjecture}
\newtheoremstyle{nonindented}{1ex}{1ex}{}{}{\bfseries}{.}{.5em}{}
\newtheoremstyle{indented}{1ex}{1ex}{\itshape\addtolength{\leftskip}{0.6cm}\addtolength{\rightskip}{0.6cm}}{}{\bfseries}{.}{.5em}{}
\theoremstyle{nonindented}
\theoremstyle{indented}
\theoremstyle{plain}

\renewcommand{\hat}{\widehat}
\renewcommand{\tilde}{\widetilde}
\renewcommand{\bar}{\overline}

\DeclareMathOperator{\poly}{poly}

\def\min{\qopname\relax n{min}}
\def\max{\qopname\relax n{max}}

\def\argmax{\qopname\relax n{argmax}}

\def\Pr{\qopname\relax n{\mathbf{Pr}}}
\def\Ex{\qopname\relax n{\mathbf{E}}}

\newcommand{\RR}{\mathbb{R}}

\def\D{\mathcal{D}}

\def\P{\mathcal{P}}

\newcommand{\mini}[1]{\mbox{$\min$} & {#1} & &\\}
\newcommand{\maxi}[1]{\mbox{$\max$} & {#1 } & & \\}

\renewcommand{\st}{\mbox{\text{s.t.}} }
\newcommand{\con}[2]{&#1 & & #2\\}
\newcommand{\qcon}[3]{&#1, & \mbox{for } #2  & #3 \\}
\newenvironment{lp}{\begin{equation}  \begin{array}{lllr}}{\end{array}\end{equation}}
\newenvironment{lp*}{\begin{equation*}  \begin{array}{lllr}}{\end{array}\end{equation*}}

\title{Selling Information Through Consulting\thanks{This work is supported by the National Science Foundation under Grant No. CCF-1718549. }}
\author{
Yiling Chen \\
Harvard University\\
{\normalsize yiling@seas.harvard.edu}\\
\and
Haifeng Xu\\
University of Virginia\\
{\normalsize hx4ad@virginia.edu}
\and 
Shuran Zheng \\
Harvard University\\
{\normalsize shuran\_zheng@seas.harvard.edu}
}
\date{}
\begin{document}


\begin{titlepage}
	\clearpage\maketitle
	\thispagestyle{empty}

\begin{abstract}
	
    We consider a monopoly information holder selling information to a \emph{budget-constrained} decision maker, who may benefit from the seller's information. The decision maker has a utility function that depends on his action and an uncertain \emph{state of the world}. The seller and the buyer each observe a private signal regarding  the state of the world, which may be correlated with each other.  The seller's goal is to sell her private information to the buyer and extract maximum possible revenue, subject to the buyer's budget constraints. We consider three different settings with increasing generality, i.e.,  the seller's signal and the buyer's signal can be independent, correlated, or follow a general distribution accessed through a black-box sampling oracle.  For each setting,  we design information selling mechanisms which are both optimal and simple in the sense that they can be naturally interpreted, have succinct representations, and can be efficiently computed. Notably, though the optimal mechanism exhibits slightly increasing complexity as the setting becomes more general,  all our mechanisms share the same format of acting as a \emph{consultant} who recommends the best action to the buyer but uses different and carefully designed payment rules for different settings.  Each of our optimal mechanisms can be easily computed by solving a single polynomial-size linear program. This significantly simplifies  exponential-size LPs solved by the Ellipsoid method in the previous work, which computes the optimal mechanisms in the same setting but \emph{without} budget limit. Such simplification is enabled by our new characterizations of the optimal mechanism in the (more realistic) budget-constrained setting.   
\end{abstract}

\end{titlepage} 

\newpage 

\section{Introduction}

Recent years have seen a growth of information markets thanks to the tremendous increase in the volume and variety of the online data sources. The information traded includes consumer information (e.g. Acxiom, Nielsen, Oracle), credit reports (e.g. Experian, TransUnion, FICO), recommendations (e.g. Waze, Tripadvisor), etc., which can be valuable to decision makers like advertisers, retailers and loan providers.

Motivated by these applications, this paper considers a very basic setting along this line and studies how a monopoly information provider can sell information optimally to a single decision maker. The decision maker's utility depends on the action he takes (e.g. to lend the money or not) and a random \emph{state of the world}, which is uncertain to the decision maker (e.g., whether the borrower will pay the debt). The information provider and the decision maker each observe a private signal regarding the state of the world. Their signals can be correlated to one another. The information provider's goal is to sell her private signal to the decision maker in a way that maximizes the revenue, knowing that the buyer will choose what is best for himself.

This problem has been formulated and studied by~\cite{Babaioff:2012:OMS:2229012.2229024}. They identified conditions under which there is an optimal one-round mechanism and provide characterizations of the optimal mechanism. In particular, they show that a one-round optimal mechanism exists in two special settings: (1) the seller's signal and the buyer's signal are independent, \emph{or} (2) the buyer is \emph{committed} to complete the entire protocol  even if aborting the protocol will give him higher utility.  
The buyer commitment assumption is not as unrealistic as it might first sound. In particular,  
 \cite{Babaioff:2012:OMS:2229012.2229024} show that buyer commitment can be obtained using the following approach: first ask the buyer to deposit a large amount of money, then run the optimal mechanism for committed buyers, and at last refund the buyer his deposit less the payment in the mechanism.  
 The deposit step enforces the buyer's commitment to the payment of the mechanism.\footnote{This is crucial in selling information since  the payment amount itself can be correlated with the information for sale. Thus the buyer can learn information upon knowing the payment amount but before deciding to purchase. In such cases, an uncommitted buyer can simply learn the payment amount and then leave because the payment already reveals useful information to him.}   
 However, the major drawback of this approach --- as also highlighted by  \cite{Babaioff:2012:OMS:2229012.2229024} --- is that both the amount of deposit and the money transfer required by the mechanism can be extremely large compared to players' expected utilities. For instance, they give a simple example with buyer utility function taking values from $[0,5]$. However, the optimal mechanism has to ask the buyer to deposit an amount of $25004$ and then either returns $0$ or returns $50000$ to the buyer, yielding an optimal expected seller revenue no more than $2$. Clearly, such a mechanism would \emph{not} be practical in reality.  


Given the above deficiencies, this work considers a more realistic setting with a \emph{budget-constrained} buyer. In particular, the buyer in our model has a budget, indicating the maximum amount he can afford. We consider both \emph{public} budget which is known to the seller a priori, and \emph{private} budget which the seller also needs to elicit in the mechanism.  In the latter case, the buyer may misreport his private budget if this is beneficial to him. 
We assume that the seller also has a budget as the maximum amount she can afford to pay the buyer (as seen from the previous example, the optimal mechanism sometimes requires the seller to pay the buyer). We consider three settings with increasing generality --- i.e., independent signals, correlated signals, and signals drawn from a general distribution accessed through a black-box sampling oracle --- and for each setting we are able to design information selling mechanisms which are  both optimal (approximately optimal when signal distributions are accessed through samples) and simple in the sense that they can be naturally interpreted, have succinct representations, and can be efficiently computed. Notably, in contrast to the optimal mechanisms by \cite{Babaioff:2012:OMS:2229012.2229024} which are formulated as exponential-size linear programs and solved by using the ellipsoid method on the dual LP, all our mechanisms can be directly computed by solving a single polynomial-size linear program thanks to our new and succinct characterization of the optimal mechanisms.  In some sense, our results can be viewed as generalizations of the results by \cite{Babaioff:2012:OMS:2229012.2229024} since when the payments in their mechanisms are smaller than the budget limits, our mechanisms (though simpler) also achieve optimality as their mechanisms whereas when the payments in their mechanisms are too large, our mechanisms are more realistic by imposing payment limits. 
Next, we elaborate our results in more details. 

\subsection{Our Results} 


When the seller's signal and the buyer's signal are independent and the buyer's budget is public, we prove that the optimal mechanism has the following simple format: (i) it asks the buyer to report his signal; (ii) based on the  report, the seller charges the buyer a fixed amount of payment; (iii) the seller maps (possibly randomly) her signal to an action recommendation  to the buyer.  The mechanism has the following properties: (1) it is individually rational (IR) and incentive compatible (IC); (2) each action recommended to the buyer is indeed his best response, which we  call the \emph{obedience} constraint; (3) the optimal mechanism can be computed directly from solving a polynomial-size linear program (LP).  This mechanism is very much like the process of consulting, during which a client (the buyer) tells his type to the consultant (the seller) and pays the consultant for receiving the best action recommendation. As a result, we term this mechanism the \emph{consulting mechanism with direct payment ($\texttt{CM-dirP}$)}.   

When the buyer budget is private, we show via an example that any mechanism with a single-round transfer \emph{cannot} be optimal even when the buyer's signal and the seller's signal are independent. Nevertheless, for independent signals, we prove that the consulting mechanism with a slightly different payment method, coined the \emph{consulting mechanism with deposit and return ($\texttt{CM-depR}$)}, achieves the optimal revenue. In particular, instead of asking for a direct payment as in $\texttt{CM-dirP}$, the $\texttt{CM-depR}$ mechanism first asks the buyer to deposit his (private) budget and then refund the buyer his budget less the payment (thus two-rounds of transfers). Besides satisfying IR, IC and obedience, the optimal $\texttt{CM-depR}$  mechanism can also be computed by a polynomial-size LP.    

For the general setting with correlated signals drawn from an explicitly given prior distribution, we prove that for both public and private budget cases, the consulting mechanism with a simple two-round transfer process, coined the \emph{consulting mechanism  with probabilistic return ($\texttt{CM-probR}$)},  achieves the optimal revenue.  The $\texttt{CM-probR}$ has the following format: (i) it asks the buyer to report his signal and deposit his budget; (ii) based on the report, the mechanism maps the seller signal (randomly) to an action recommendation together with an amount of refund, which is either $0$ (i.e., no refund) or the buyer's full deposit plus the seller's full budget.  
The main difference between $\texttt{CM-probR}$ and the previous consulting mechanisms is that the payment amount is no longer pre-determined before any information disclosure, but is contingent upon the recommended action and thus contains information regarding the seller's private signal. Similarly, besides satisfying IR, IC and obedience, the optimal $\texttt{CM-probR}$  mechanism can also be computed by a polynomial-size LP. 
 Our proof of the optimality and the computability of $\texttt{CM-probR}$  mechanisms features a novel use of the duality theory, which may be of independent interest.  Our starting point is an LP formulation 
 $\P$ for computing the optimal mechanism which however has exponentially many constraints \emph{and} exponentially many variables. 
 We simplify this LP as follows. We first turn to the dual LP $\D$ and leverage its special structure to eliminate a majority of the dual variables and convert $\D$  into another linear program $\D'$ with only  polynomially many variables though still exponentially many constraints. 
 At this point, we can employ the celebrated ellipsoid method to solve $\D'$, which however is neither practically efficient nor able to give us an interpretable mechanism.  We instead turn to the dual of $\D'$, denoted as $\P'$, which is different from $\P$. Interestingly, despite that $\P'$ has exponentially many variables, we manage to prove that its solutions correspond to a special class of information selling mechanisms and moreover, always admits an optimal solution which is a $\texttt{CM-probR}$. This yields the optimality of  $\texttt{CM-probR}$  mechanisms.
 
 Finally, we extend our results to the setting where the seller does not know the prior distribution of the signals explicitly and can only access it through a black-box sampling oracle. This applies to the cases where distributions are only accessible through samples. Using a Monte-Carlo sampling approach, inspired by the idea of \cite{dughmi2016algorithmic}, we show that a mechanism with  approximate optimality, IR, IC and obedience can be implemented in polynomial time with polynomially many samples from the prior distribution.    

\subsection{Related Work}
Most relevant to this paper is the work by \cite{Babaioff:2012:OMS:2229012.2229024} who study the optimal mechanism design for selling information to an imperfectly informed decision maker (henceforth the \emph{buyer}) without budget consideration. 
Our setting is essentially the same as that of \cite{Babaioff:2012:OMS:2229012.2229024}, but under \emph{budget constraints}. This is motivated by the drawbacks of their optimal mechanisms  which may require an extremely large amount of transfers even to extract a tiny amount of revenue. On the descriptive side,  we show that budget constraints indeed affect the characterization of optimal mechanisms. For example, when the buyer has a private budget constraint, any mechanism with one-round transfer cannot be optimal even when signals are independent, which is opposed to the characterization in \cite{Babaioff:2012:OMS:2229012.2229024} for the same setting but without budgets.  On the prescriptive side, all the mechanisms in \cite{Babaioff:2012:OMS:2229012.2229024} rely on solving exponentially large linear programs via the ellipsoid method whereas all our mechanisms can be computed  by directly solving a polynomial-size linear program. This is enabled by our new characterizations of the optimal mechanism in much simpler and more interpretable forms. 

The sale of information has attracted much attention in the economic literature.  To our knowledge,  \citet{admati1986monopolistic, admati1990direct} are among the first to explicitly consider the sale of information. They consider a monopoly who sells information regarding a risky asset to a continuum of  homogeneous traders in a financial market, which then affects the equilibrium price of the market. They show that the optimal selling method may involve revealing different and noisy signals to different traders. In contrast to the focus of \citet{admati1986monopolistic, admati1990direct}  on selling information to multiple homogeneous traders with externalities,  we study the sale of  information to a single decision maker with heterogeneous types.  Recently, \cite{bergemann2015selling} study the sale of cookies, a particular type of informational goods, motivated by the context of online advertising.  \cite{malenko2018proxy}  studies how to sell information to voters and characterize when the sale of information leads to a more informative voting outcome.  \cite{Bergemann2018Info} considers the sale of information to a decision maker, which is similar to our setting  but without budget constraints.  Moreover, they restrict the design space to a particular format of mechanisms as a menu of statistical experiments. They establish various properties regarding the optimal mechanism and characterize the optimal mechanism in the cases of binary states and actions, or binary types.  In particular, part of our proof for the characterization results of independent signals uses similar ideas by \citet{Bergemann2018Info} to reduce the number of signals required in the signaling scheme. { \citet{esHo2007price} studied the  consultant's problem of optimally contracting with the client. 
	They considered a different model in which a consultant incurs a cost to acquire a private observation, and once the observation is acquired, the consultant will fully reveal the information to the client.}  

{ Selling physical goods to buyers with budget constraints has been studied extensively in the literature of mechanism design (see, e.g., \citep{che2000optimal, chawla2011bayesian, devanur2017optimal} and references therein). }
Our work is conceptually related to, yet fundamentally different from, selling physical goods and meanwhile revealing information regarding the value of the item for sale.  \citet{Emek12,Miltersen12,Badanidiyuru2018} study how an auctioneer can  strategically reveal information to bidders in order to affect their valuation regarding the item with the ultimate goal of maximizing revenue in a fixed auction format (in particular, the second-price auction).  \cite{Daskalakis2016does} consider  a seller of a physical item who also possesses private information regarding the item's value to buyers, and studies the optimal mechanism for jointly selling the item  and the seller's private information. They show that this joint design reduces to optimal mechanism design for selling multiple items without information involved.  \cite{smolin2019disclosure} considers a similar situation where a seller sells an item, characterized by a vector of attributes, to a buyer. They also study the optimal joint design of the item pricing mechanism and information selling mechanism and provide a characterization of the optimal mechanism within a particular design space.

Finally, our work is also related to the rich literature of persuasion, a.k.a. \emph{signaling} or \emph{information design}. Our model is particularly relevant to persuasion of a privately informed receiver by \cite{kolotilin2017persuasion} since the seller in our model reveals information by a signaling scheme to a buyer with a private signal. The techniques we use in Section \ref{sec:black-box} is inspired by algorithms for Bayesian persuasion designed by \cite{dughmi2016algorithmic}. 
Due to the space limit, we do not give a thorough review of the large literature in persuasion but refer the reader to the recent survey by \cite{kamenica2018bayesian} (from the economic perspective) and \cite{Dughmi2017} (from the algorithmic perspective).     
 




\section{Preliminaries}\label{sec:model}
\paragraph{Basic Setup}
We consider a monopoly information holder (call \emph{her} the \emph{seller}) selling information to a budget-constrained decision maker (call \emph{him} the \emph{buyer}), who needs to choose an action $a \in A$ and has a budget $b\in B$.  The buyer's utility  depends  on both his action $a$ and a random \emph{state of the world}. The seller and the buyer each observe a private \emph{signal} about the state of the world, denoted by $\omega \in \Omega$ and $\theta \in \Theta$, respectively.  Sets $A,B,\Omega,\Theta$ are all finite. Let $u(\omega, \theta, a)$ denote the buyer's utility of action $a$ when the seller's signal and buyer's signal are realized to $\omega, \theta$, respectively.\footnote{Alternatively and equivalently, one can think of the buyer's utility as a function $u'$ of action $a$ and the state of world denoted as  by $s$.  Then our definition of $u(\omega, \theta, a) $ simply corresponds to $\E_s[u'(s,a) | \omega, \theta].$ } For convenience, we sometimes also call $\omega$ and $\theta$ the \emph{type} of the seller and buyer. Moreover, we assume that the seller also has a budget $M$ to run the mechanism and $M$ is \emph{publicly known} (as we shall see, the optimal mechanism sometimes requires the seller to pay the buyer with some probability). 
 

We assume that the signal pair $\omega, \theta$ and the budget $b$ are drawn from a publicly known prior distribution $\mu(\omega, \theta, b)$ and the utility function $u(\omega, \theta, a)$ is also public knowledge. The surplus for a buyer of type $\theta$ from fully observing the seller's signal  $\omega$, denoted as $\delta(\theta)$, is
$$
\delta(\theta) = \E_\omega\left[ \max_a u(\omega, \theta, a)\big\vert \theta\right] - \max_{a\in A} \E_\omega\left[u(\omega, \theta, a) | \theta\right] \geq 0.
$$

The seller's goal is to sell her information regarding $\omega$ to the budget-constrained buyer to extract the largest possible revenue. As is typical in mechanism design, we assume that the seller is a trusted authority and will  not defect after posting a selling mechanism based on her knowledge of $\mu(\omega, \theta, b)$ and $u(\omega, \theta, a)$. However, the buyer will strategically respond to the mechanism in order to maximize his own utility and may defect if that is beneficial.  We remark that this basic model has been studied in some previous works \cite{Babaioff:2012:OMS:2229012.2229024,Bergemann2018Info}. The key difference between our model and previous ones is that  our seller and buyer are both \emph{budget-constrained}.

\paragraph{The Design Space -- Generic Interactive Protocols.} We now describe the design space, i.e., the class of mechanisms within which our mechanism will be optimal. A key observation by~\cite{Babaioff:2012:OMS:2229012.2229024}  is that a mechanism that interacts with the buyer for \emph{multiple} rounds can possibly extract more revenue than any one-round mechanism. Therefore, we consider mechanisms that may interact with the buyer for multiple rounds via a \emph{generic interactive protocol}, as defined by \cite{Babaioff:2012:OMS:2229012.2229024}.  At a high level, the protocol/mechanism can be viewed as an extensive-form game represented as a tree with three types of nodes: \emph{seller} nodes, \emph{buyer} nodes and \emph{transfer} node.   Each move at a seller node can be seen as revealing some amount of information regarding $\omega$ to the buyer whereas 
each move at a buyer node reveals buyer information (e.g., the buyer's type) to the seller. All the payments ---  from the buyer to the seller or vice versa --- happen at transfer nodes.  More formally, the protocol is defined as follows.  


\begin{definition}[\cite{Babaioff:2012:OMS:2229012.2229024}]
A \emph{generic interactive protocol} is a finite-size
protocol tree with the following three types of non-leaf nodes:
\begin{enumerate}
    \item {\bf Seller node}: at a seller node $n$, the seller randomly moves to one of the children of $n$ according to her private signal $\omega$. Let $C(n)$ be any child of $n$ and  $p_n(\omega, C(n))$ be the probability of moving to $C(n)$ when the seller's signal is $\omega$. Note that $\sum_{C(n)} p_n(\omega, C(n)) = 1$ for all $\omega$. 
    \item {\bf Buyer node}: at a buyer node $n$, the buyer (with type $\theta$ and budget $b$) randomly moves to one of the children $C(n)$. A buyer's strategy at node $n$ is a transition function $\phi_n((\theta, b),C(n))$, which indicates the probability of moving to the node $C(n)$. Naturally,  it should satisfy $\sum_{C(n)} \phi_n((\theta,b),C(n)) = 1$.
    \item {\bf Transfer node}: any transfer node $n$ only has one child and an associated payment to the seller of amount $t(n)$.  
    Note that $t(n)$  can be  negative, meaning the seller pays the buyer.  
\end{enumerate}
\end{definition}

The whole protocol (including all the transition probabilities) is publicly known and both players know which node they are currently at during the execution.  However, each player's private type is only known to themselves. Nevertheless, the buyer at any buyer node can form a posterior belief on $\omega$ according to what has been observed so far along the tree path based on the transition function $p_n(\omega, C(n))$ at each seller node. 
During the game, the buyer can choose to leave the protocol if that is beneficial or the total spending exceeds his budget. Otherwise, the game terminates at some leaf node of the protocol. In either case, the buyer will then choose the best action based on his posterior belief of $\omega$. 
{ Furthermore, the seller's budget constraint restricts the net transfer from seller to buyer to never exceeds $M$ at any node of the protocol tree, and the buyer will never choose to visit a node of the protocol tree at which the net transfer from buyer to seller exceeds $b$.}

We remark that all the optimality claims throughout the paper will be within the design space of all generic interactive protocols, though we  tend to omit this emphasis for convenience.  
We first simplify the design space by invoking the revelation principle. 
\begin{lemma}[Revelation principle] \label{lem:revelation}
	There exists an optimal generic interactive protocol in which the buyer truthfully reports type $\theta$ and budget $b$ at the beginning and takes no other actions. Therefore, the protocol only has one buyer node, i.e., the root of the protocol tree. We call such a mechanism a \emph{revelation mechanism}.  
\end{lemma}
The proof can be found in Appendix~\ref{app:revelation}.
A revelation mechanism is  \emph{incentive compatible} (IC) if it is an optimal strategy for the buyer to report his true type $\theta$ and budget $b$ at the beginning, and is  \emph{individually rational} (IR) if the expected  utility of participating the mechanism is non-negative for any buyer. 

\paragraph{Information Revelation via Signaling Schemes.} Though the generic interactive protocols reveal information sequentially in general and can be intricate, we prove that  optimal mechanisms will only need to reveal information once (though transfers may happen at multiple rounds)\footnote{ If one imposes a no-positive-transfers constraint, that $t(n)$ cannot be negative at any transfer node, then \cite{Babaioff:2012:OMS:2229012.2229024} present an example where the seller may need to reveal information more than once in the optimal mechanism.}, and such information revelation can be described via a \emph{signaling scheme}.  Concretely, a signaling scheme is a \emph{randomized}  mapping from the seller's signal set $\Omega$ to a set of signals $\Sigma$, which can be fully described by the likelihood function $\{ p(\omega, \sigma) \}_{\omega \in \Omega, \sigma \in \Sigma}$ where $p(\omega, \sigma)$ is the probability of sending signal $\sigma$ given seller signal $\omega$.   



With slight abuse of notation, we use $\mu(\omega)$ to denote the buyer's  belief regarding the seller's signal $\omega$. After receiving any signal $\sigma$ from the signaling scheme, the buyer with prior belief $\mu$ will update his belief via a standard Bayesian update, and infer that the buyer signal is $\omega$ with probability 
\begin{equation}\label{eq:buyer-update}
\Pr(\omega| \sigma) =\frac{ \mu(\omega) p(\omega, \sigma)  }{  \sum_{\omega' \in \Omega}  \mu(\omega') p(\omega', \sigma)  }. 
\end{equation} 

As an example, each buyer node $n$ in the generic interactive protocol, together with its transitions described by $\{ p_n(\omega, C(n))  \}_{\omega, C(n)}$, is equivalent to a signaling scheme in which each $C(n)$ corresponds to a signal and  $p_n(\omega, C(n))$ is the probability of sending signal $C(n)$ given the seller type $\omega$. Under this interpretation, moving to node $C(n)$ is equivalent to that the seller sends a signal. The buyer then observes the signal and updates his belief regarding $\omega$ according to Equation \eqref{eq:buyer-update}. Though such information revelation and buyer's belief updates may happen multiple times in the generic interactive protocol, we will show that it suffice to reveal information only once in the optimal mechanism.


\paragraph{Player Utilities and the Revenue Maximization Problem.} 

Let  $Z_{\theta,b}$ denote the (random) node that the protocol ends at (so $Z_{\theta, b}$ could be a leaf, or a node where the buyer choose to leave) when interacting with a buyer of type $\theta$ and budget $b$. Let $\mu(\omega|Z_{\theta,b})$ denote the buyer's posterior belief probability of $\omega$ at $Z_{\theta, b}$ and $t(Z_{\theta, b})$ be the summed transfer on the path from the root to $Z_{\theta,b}$. The expected utility of a buyer of type $\theta$ and budget $b$ is then equal to 
\begin{align*}
U(\theta,b) = \E_{Z_{\theta,b}} \left[\max_a  \sum_{\omega\in \Omega} [u(\omega, \theta, a) \cdot \mu (\omega|Z_{\theta,b}) - t (Z_{\theta,b}) \right],
\end{align*}
where the randomness of node $Z_{\theta,b}$ is due to the internal randomness of the mechanism. Note that the buyer can always leave at the beginning, with expected utility
$\max_a \  \sum_{\omega \in \Omega } u(\omega, \theta, a) \cdot \mu(\omega| \theta,b)$.

On the other hand, the seller's revenue of a generic interactive protocol $M$ is  
\begin{align*}
\texttt{Rev}(M) = \sum_{\theta, b}  \mu(\theta,b) \cdot \E_{Z_{\theta,b}} \left[ t(Z_{\theta,b})  \right],
\end{align*}
where $\mu(\theta, b) = \sum_{\omega} \mu(\omega, \theta, b)$ is the prior probability of a buyer of type $\theta$ and budget $b$ showing up. The seller's goal is to choose the optimal protocol $M$ that maximizes the above expected revenue, assuming  a rational buyer who will optimize his own utility in this mechanism. 

\section{Independent Signals}\label{sec:indep} 
An interesting special case of the problem is that the buyer's signal and budget  $(\theta, b)$  is independent from the seller signal $\omega$, i.e., $\mu(\omega, \theta, b) = \mu(\omega) \cdot \mu(\theta, b)$ where $\mu(\cdot)$ denotes the prior distribution of the corresponding random variables.  
We study this special case separately due to two reasons. First,  this setting has another natural interpretation --- i.e., the $\theta$ can also be interpreted as  the buyer's private type that captures the buyer's \emph{information-irrelevant} characteristics (thus independent of the seller's signal)  such as, e.g., properties of his decision problem $u$.   Second,  the optimal mechanism of this case has a more concise representation. 
In particular, we prove that in this case, there exists an optimal budget-feasible mechanism with the following simple procedure: (1) it asks for the buyer's type $\theta$ and budget $b$; (2) charges the buyer a fixed amount  (via direct payment or first ask for a deposit of $b$ and then return a part of it); (3) recommend an action for the buyer to take.  This mechanism is similar in spirit to a ``consulting procedure'' during which the buyer (a client) tells his type to the seller  (the consultant). The seller then  charges the buyer a fixed amount of money and recommends an action to the buyer to take based on his reported type. 
We thus call such mechanisms  \emph{consulting mechanisms}. It turns out that, with independent signals, the optimal mechanism needs to use slightly different payment methods for the public and private budget cases. We thus consider them separately in subsection \ref{sec:indep:pub} and \ref{sec:indep:priv}.    



\subsection{Public Budget}\label{sec:indep:pub}

In this subsection, we prove that when the signals are independent and the buyer has a fixed public budget $b$, there always exists an optimal consulting mechanism with one-round direct  payment from the buyer to the seller, formally defined as follows. 
\begin{definition} [Consulting mechanism with direct payment]
	A \emph{consulting mechanism with direct payment ($\texttt{CM-dirP}$)} for a buyer with publicly known budget $b$ proceeds as follows: 
	\begin{enumerate}
		\item The seller \emph{commits} to a payment amount $t_{\theta}(\leq b)$ and an action recommendation policy described as a  mapping $p_{\theta}: \Omega \to \Delta A$ for each buyer type $\theta$, where $p_{\theta}(\omega) \in \Delta A$  with $p_{\theta}(\omega, a)$ denoting the probability of recommending action $a$ given seller signal $\omega$.  $p_{\theta}$ is required to be \emph{obedient} --- i.e., conditioned on any recommended action, it  must indeed be an optimal action for the buyer given his information. 
		\item The seller asks the buyer to report his type $\hat{\theta}$. 
		\item The seller charges the buyer an amount $t_{\hat{\theta}}$.
		\item Based on her  signal $\omega$, the seller samples an action $a \sim p_{\hat{\theta}}(\omega)$ and recommends $a$ to the buyer. 
	\end{enumerate} 
\end{definition}
We make a few remarks about $\texttt{CM-dirP}$. First, $\texttt{CM-dirP}$ is indeed a generic interactive protocol, described as follows: (1) its root is  the only buyer node at which he reports his type $\hat{\theta}$; (2) each child of the root corresponds to a buyer type $\hat{\theta}$ and is a transfer node  with transfer amount $t_{\hat{\theta}}$ (to the seller); (3) following each transfer node is a seller node whose children are leaves, each indexed by a buyer action $a \in A$, and the transition probability to leaf node $a \in A$ is $p_{\hat{\theta}}(\omega, a)$. Second, as described in Section \ref{sec:model}, the transition from seller node to leaves is effectively a signaling scheme, in which the set of signals coincides with the set of actions. Upon receiving an action recommendation $a$, the buyer infers a posterior probability of $\omega$ by Bayes updates, as follows: $\Pr(\omega|a) = \frac{ \mu(\omega) p(\omega, a)  }{  \sum_{\omega' \in \Omega}  \mu(\omega') p(\omega', a)  }$. Finally,  an intrinsic constraint of $\texttt{CM-dirP}$  is that the seller's action recommendation must satisfy the \emph{obedience constraint}. That is,  the recommended action must indeed be an optimal action for the buyer's  reported type or, more formally, $$a  = \argmax_{a' \in A} \sum_{\omega} u(\omega, \hat{\theta}, a') \cdot \frac{ \mu(\omega) p_{\hat{\theta}}(\omega, a)  }{  \sum_{\omega' \in \Omega}  \mu(\omega')  p_{\hat{\theta}}(\omega', a)  }. $$ 

Our main result of this subsection is to prove the optimality of $\texttt{CM-dirP}$ mechanisms, described as follows. The proof of Theorem \ref{thm:consult-public-opt} is deferred to Appendix~\ref{sec:public}.

\begin{theorem}\label{thm:consult-public-opt}
When the buyer budget $b$ is public and the buyer signal $\theta$ is independent from $\omega$,  there always exists an optimal mechanism that is an IC \emph{consulting mechanism with direct payment ($\texttt{CM-dirP}$)}.   
\end{theorem}

\subsection{Private Budget}\label{sec:indep:priv}
When the budget  is privy to the buyer, our mechanism will also need  to ask the buyer to report his budget. 
It turns out that in this case, any mechanism with single-round transfer  \emph{cannot} achieve optimality, which is illustrated in the following example. 
\begin{example}\label{ex:box}
We consider a similar setup as the``treasure box'' example used  by~\cite{Babaioff:2012:OMS:2229012.2229024}. In particular, imagine a box with a locker on it. There are two keys labeled with $0$ and $1$, exactly one of which can open the box. The buyer can choose one key and try it. If he opens the
box, he gets the object inside. Let the type of the buyer $\theta \in \{ 0 , 1\}$ encode his value $z_\theta$ for the object, where $z_0 = 120$ and $z_1 = 80$.   The seller knows exactly what
the correct key is, and let it be $\omega \in \{ 0, 1\}$. How should the seller sell her information to the buyer?

In this problem, we  have $\Omega =\Theta = A = \{0, 1\}$, $u(\omega, \theta, a) = \mathds{1} \{\omega = a\} \cdot z_\theta$.  A type-$0$ buyer has a budget of $50$ whereas a type-$1$ buyer has a budget of $100$.   The signals $\theta$ and $\omega$ are independent and uniformly at random. Concretely, the prior distribution can be expressed as follows: $\mu(\omega, \theta, b) = 1/4$ when $(\omega, \theta, b) = (0, 0, 50), (1,0,50),(0, 1, 100), (1, 1, 100)$ and $\mu (\omega, \theta, b) = 0$ otherwise. 
Simple calculation shows that the buyer's surplus from fully knowing the seller's signal  $\omega$ will be  $ \delta(0) = 60$ and $ \delta(1) = 40. $

It can be shown that some  $\texttt{CM-dirP}$ mechanism achieves optimality among all mechanisms \emph{with a single-round of transfer}. Moreover,  the optimal $\texttt{CM-dirP}$ mechanism is to charge a fixed price $40$ and then reveal the exact value of $\omega$, which achieves  expected revenue $40$.\footnote{in Section~\ref{sec:indep:comp}, we will show how to compute the optimal $\texttt{CM-dirP}$ mechanism via a  polynomial-size linear program.} 
However, strictly more revenue can be extracted via the following mechanism with \emph{two rounds of transfers}. 
Consider the mechanism that offers the buyer the following two options: 
\begin{enumerate}
\item Pay $50$ dollars to know the exact value of $\omega$.
\item Pay $100$ dollars first and then get a refund of $61$ dollars together with the exact value of $\omega$.
\end{enumerate}
The type-$0$ buyer only has $50$ dollars, so he can only choose Option 1. The type-$1$ buyer will choose Option 2.
The expected revenue of this mechanism is thus $(50+39)/2 = 44.5$, which is higher than $40$.
\end{example}

We show that consulting mechanisms can still achieve optimality but with a slightly different payment method: first ask the buyer to fully deposit his budget, and then return a part of the deposit back.  
\begin{definition}[Consulting mechanism with deposit and return]
	A \emph{consulting mechanism with deposit and return ($\texttt{CM-depR}$)}  proceeds as follows: (1) The seller commits to a payment amount  $t_{\theta,  b}(\leq b)$ and an action  recommend policy  $p_{\theta, b}: \Omega \to \Delta A$ for each $(\theta,b)$, under \emph{obedience} constraints (similar to Step 1 of the $\texttt{CM-dirP}$ mechanism); (2) Asks the buyer to report his type $\hat{\theta}$ and deposit his budget $\hat{b}$; (3) Returns to the buyer an amount of $\hat{b} - t_{\hat{\theta}, \hat{b}}$; (4) Samples an action $a \sim p_{\hat{\theta},\hat{b}}(\omega)$ and recommends $a$ to the buyer.  
\end{definition}

As can be seen, a $\texttt{CM-depR}$ mechanism is almost the same as $\texttt{CM-dirP}$ except that it charges the buyer $t_{\hat{\theta},\hat{b}}$ by first asking him to deposit the budget $\hat{b}$ and then immediately returning $\hat{b}-t_{\hat{\theta},\hat{b}}$   whereas  $\texttt{CM-dirP}$ just asks for a direct payment $t_{\hat{\theta}}$. This step is useful in the private budget case because it serves as a ``verification'' of the buyer's budget which helps the seller to learn more information about the buyer's type and thus increases her power of price discrimination.  This intuition is also illustrated  in Example \ref{ex:box}. 

\begin{theorem}\label{thm:consult-private-optimal}
When the buyer  has a private budget and $(\theta, b)$ is independent from $\omega$,  there always exists an optimal mechanism that is an IC \emph{consulting mechanism with deposit and return ($\texttt{CM-depR}$)}.  
\end{theorem}

The proof can be found in Appendix~\ref{sec:private}.

\subsection{Computing Optimal Consulting Mechanisms} \label{sec:indep:comp}
We now show how to efficiently compute the optimal consulting mechanisms for both public and private buyer budget. It turns out that in both cases, the optimal consulting mechanisms can be easily computed by simple linear programs with polynomial size. This much simplifies the optimal mechanism  proposed by \cite{Babaioff:2012:OMS:2229012.2229024} (for the setting \emph{without} budget constraints), which requires solving exponentially-large linear programs by going through the dual program and the  ellipsoid method.  Our mechanism is simpler, more interpretable and can be more efficiently computed from a practical perspective. 


Here, we only give the LP formulation for computing the optimal $\texttt{CM-depR}$ mechanism since  the formulation for the optimal $\texttt{CM-dirP}$ is essentially the same. By definition, a $\texttt{CM-depR}$ mechanism can be fully described by variables $p_{\theta,b}(\omega, a)$,  which is the probability of recommending action $a$ when the seller signal is realized to $\omega$ and the buyer reports $(\theta, b)$, and the corresponding net transfer variable $t_{\theta, b}$.  If a $(\theta, b)$-buyer misreports type $\theta'$ and budget $b'$, he will receive a recommendation generated according to the random mapping $p_{\theta',b'}(\omega, a)$ from the seller. However, this recommended action does not have to be optimal for this buyer due to his misreport. Nevertheless, the buyer's optimal expected utility can still be computed as follows, which is the sum of the optimal expected utilities from all action recommendations: 
$$
\sum_{a} \max_{a'} \sum_\omega \mu(\omega)  p_{\theta', b'}(\omega, a)  u( \omega, \theta, a') - t_{\theta', b'} .
$$

 As a result, the optimal consulting mechanism can be computed via the following linear program (LP) with $p_{\theta,b}(\omega, a)$'s and $t_{\theta,b}$'s as variables.  


\begin{lp}
    \maxi{ \sum_{\theta, b} \mu(\theta, b) \cdot t_{\theta, b} }\\
    \st
    \qcon{ \sum_\omega \mu(\omega)  p_{\theta, b}(\omega, a)  u(\omega, \theta, a) \ge \sum_\omega \mu(\omega) p_{\theta, b}( \omega, a) u(\omega, \theta, a')}{a, a',  \theta, b}{\text{(OB)}}
    \con{ \sum_{\omega, a} \mu(\omega)  p_{\theta, b} (\omega, a)  u(\omega, \theta, a) - t_{\theta, b} }{}
    \qcon{ \qquad \ge \sum_{a} \max_{a'} \sum_\omega \mu(\omega)  p_{\theta', b'}(\omega, a)  u( \omega, \theta, a') - t_{\theta', b'}}{ \theta, \theta',  b, b' }{\text{(IC)}}    
     \qcon{ \sum_{\omega, a} \mu(\omega)  p_{\theta, b} (\omega, a) u(\omega, \theta, a) - t_{\theta, b} \ge \sum_\omega \mu(\omega) u(\omega, \theta, a')}{ (\theta, b), a' }{\text{(IR)}}
    \qcon{ t_{\theta, b} \le b}{ (\theta, b) }{\text{(budget)}}
    \qcon{ \sum_a p_{\theta, b}(\omega, a) = 1}{ (\theta, b), \omega \notag}{}
    \con{ 0 \le p_{\theta, b}(\omega, a) \le 1}{}
\end{lp}


The first constraint guarantees the \emph{obedience} constrain in the consulting recommendation, i.e., each recommended action  must indeed be an optimal action for the buyer's  reported type.  The second constraint guarantees incentive compatibility (IC) and the third constraint guarantees individual rationality (IR).

We remark that the IC constraints above are not linear (yet). However, one  can easily transform them into linear constraints by introducing new variables $z_{\theta, \theta', b', a} $ to represent the $\max$ function, as follows 
$$
z_{\theta, \theta', b', a} = \max_{a'} \sum_\omega \mu(\omega) \cdot p_{\theta', b'}(\omega, a) \cdot u( \omega, \theta, a').
$$
Then for any $ \theta, \theta',  b \ge b' $, the IC constraint can be re-formulated as  follows:
\begin{align*}
\sum_{\omega, a} \mu(\omega) \cdot p_{\theta, b} (\omega, a) \cdot u(\omega, \theta, a) - t_{\theta, b} 
\ge  \sum_{a} z_{\theta, \theta', b',a}  - t_{\theta', b'}
\end{align*}
with additional constraints that define $z_{\theta, \theta', b', a}$:
\begin{align*}
   z_{\theta, \theta', b', a} \ge \sum_\omega \mu(\omega) \cdot p_{\theta', b'}(\omega, a) \cdot u( \omega, \theta, a'), \quad \forall a' \in A. 
\end{align*}
These overall establish the following theorem. 
\begin{theorem}
The optimal  $\texttt{CM-dirP}$ mechanism for public buyer budget and the optimal $\texttt{CM-depR}$ mechanism for private buyer budget can each be computed by a single linear program of $\poly(|\Theta|, |A|, |\Omega|)$ size.  
\end{theorem}

\section{Correlated Signals}
In this section, we turn to the general setting with correlated signals. In this case, buyers of different signal $\theta$'s will have different prior beliefs on $\omega$ due to the correlation between $\omega$ and $\theta$. On one hand, this increases the seller's power to do price discrimination. On the other hand, it also complicates the design of the optimal mechanism. In particular, with independent signals, the payments of the optimal mechanisms only depend on buyer types but do not need to be contingent on the information  revealed (recall that  $\texttt{CM-dirP}$ and $\texttt{CM-depR}$ ask for payments \emph{before} any information is revealed). However, \cite{Babaioff:2012:OMS:2229012.2229024} show that this \emph{ceases to hold} for correlated signals, even when there is no budget constraint. That is, when the seller's signal and the buyer's signal are correlated, the payment of the optimal mechanism has to depend on the information that is revealed to the seller.   

In this section, we prove that the consulting mechanism coupled with a particular way of payments (which does depend on the information revealed) ---   together termed the \emph{Consulting Mechanisms with Probabilistic Return ($\texttt{CM-probR}$) } --- is optimal for both the public and private buyer budget.  With this characterization result, we then show that the optimal mechanism can be computed efficiently, again by formulating a polynomial-size linear program to directly solve for the optimal $\texttt{CM-probR}$ mechanism.  




\subsection{Optimality of $\texttt{CM-probR}$ for Correlated Signals }  

We start by formally defining the Consulting Mechanisms with Probabilistic Return ($\texttt{CM-probR}$), as follows

\begin{definition}[Consulting Mechanism with Probabilistic Return]\label{def:BDR}
	Let $[I] = \{ +, -\}$ contain the \emph{indicators} about whether the buyer will receive a return (``$+$'') or not (``$-$''). A \emph{consulting mechanism with probabilistic return ($\texttt{CM-probR}$)} for a budget-constrained buyer proceeds as follows:
	\begin{enumerate}
		\item The seller commits to an action recommendation policy described as a  mapping $p_{\theta, b}: \Omega \to \Delta( A\times I)$ for each $(\theta,b)$, where $p_{\theta, b}(\omega) \in \Delta( A\times I)$ with $p_{\theta,b}(\omega,[a,i])$ denoting the probability of recommending action $a$ with indicator $i \in \{ +,- \}$ given seller signal $\omega$.  $p_{\theta, b}$ is required to be \emph{obedient} --- i.e., conditioned on any recommended action, it must indeed be the buyer's optimal action given his information. 
		\item The seller asks the buyer to report his type $\hat{\theta}$ and  deposit his budget $\hat{b}$. 
		\item Based on her signal $\omega$, the seller samples $[a,i] \sim p_{\hat{\theta},\hat{b}}(\omega)$, and recommends $a$ to the buyer.
		\item If  the sampled $i =$ ``$-$'', return $\hat{b}+M$ to the buyer; return $0$ otherwise (recall $M$ is the seller's budget).   
	\end{enumerate}  
\end{definition}

The key difference between $\texttt{CM-probR}$ and $\texttt{CM-dirP}$ lies in their ways of payments. In $\texttt{CM-probR}$, the seller has some probability of receiving a return of $\hat{b}+M$ and this probability depends on the information revealed, which is carried by $[a,i]$ pair. One might wonder why we treat the two cases $+$ and $-$ (w.r.t. an action $a$) separately while not simply use the expected payment instead. We note that this would \emph{not} work because  the posterior probability of each payment is different for different buyer types due to their different prior beliefs  on $\omega$. The combined payment will \emph{not} simply translate to the same expectation for different buyer types due to certain nonlinearity, which will be problematic for IC constraints.  


Our main result of this subsection is  to prove the optimality of $\texttt{CM-probR}$ for correlated signals. 
\begin{theorem}\label{thm:BDR-opt}
When $\theta, \omega$ are correlated, there always exists an IC  Consulting Mechanism with Probabilistic Return ($\texttt{CM-probR}$) that maximizes the seller's revenue in both public and private buyer budget settings.  
\end{theorem}


The proof of Theorem \ref{thm:BDR-opt}  relies on a novel use of the duality theory. We defer the full proof to Appendix \ref{app:corr_lem}  and only provide a sketch here. Our starting point is a result by \cite{Babaioff:2012:OMS:2229012.2229024} who proved that for the setting \emph{without} budget constraints, there always exists an optimal mechanism of the following format: (1) ask the buyer for a (possibly extremely large) deposit; (2) reveal information through a \emph{signaling scheme}; (3) return to the buyer some amount of the deposit which will depend on the signal sent. They term it a \emph{pricing outcomes mechanism}. As a first step, we generalize this result and prove that it still holds even with budget constraints, but with Step (1)  now being to ask the buyer to deposit his reported budget. 

Observe that $\texttt{CM-probR}$ mechanisms are a strict \emph{sub-class} of pricing outcomes mechanisms --- in particular, $\texttt{CM-probR}$'s use at most $2n$ signals and only two possible payments, i.e., returning $b+M$ or $0$.    
We want to prove that this much smaller class of mechanisms is still able to achieve the optimal revenue. Our idea goes as follows (also see illustrations in Figure \ref{fig:cor-alg-explain}). We first formulate the problem of computing the optimal pricing outcomes mechanism, which is a linear program $\P$ with \emph{exponentially} many variables and \emph{exponentially} many constraints. Naturally, any naive approach ---  including  the celebrated ellipsoid method --- cannot be directly used to solve such an LP or its dual. We nevertheless work with the dual of $\P$, denoted as LP $\D$. Obviously, $\D$ also has exponentially many variables and exponentially many constraints, which is still difficult to solve.  
\begin{figure}[ht]
	\centering
	\includegraphics[clip,scale=0.4]{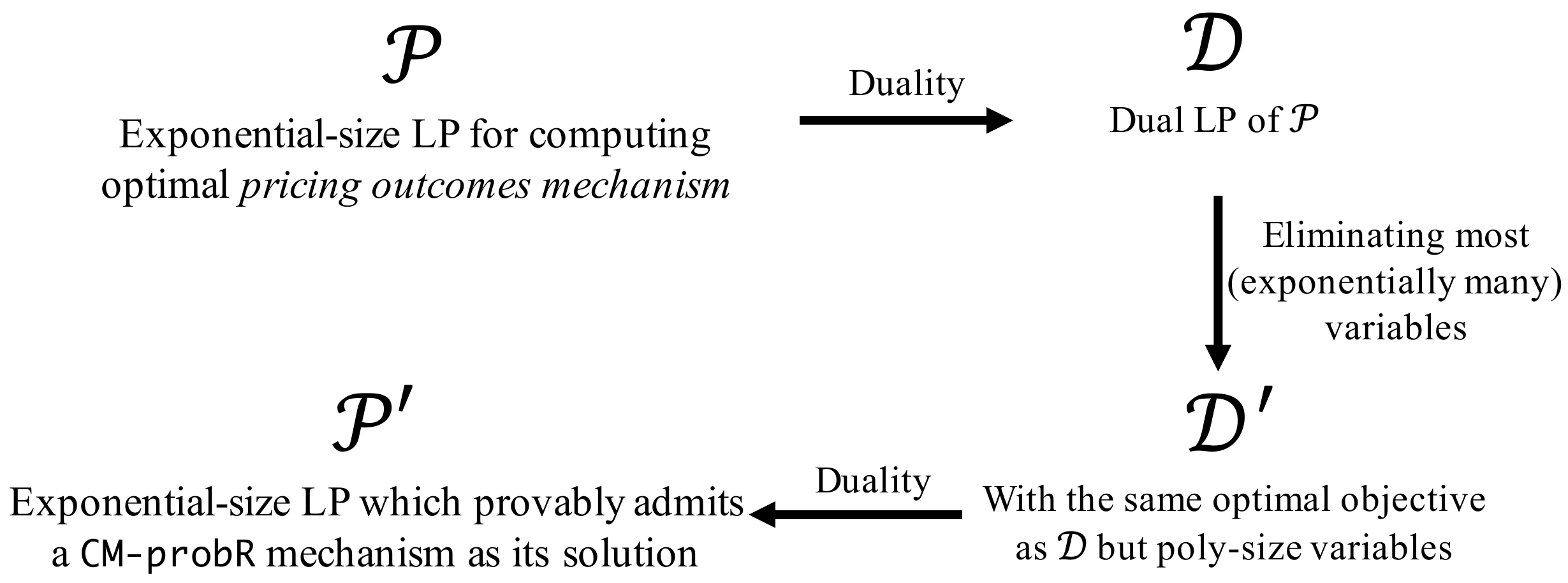}
	\caption{\label{fig:cor-alg-explain} Illustration of the Main Idea for Proving Optimality of $\texttt{CM-probR}$.}
\end{figure}

The crux of  our approach is to prove that the dual LP $\D$ can be reduced to another linear program $\D'$, which will achieve the same optimal objective value as of LP $\D$,  still has exponentially many constraints but will have only polynomially many variables. At this point, one natural idea is to apply the ellipsoid method to solve $\D'$ to obtain the optimal solution, which turns out indeed is doable in polynomial time by designing a separation oracle for $\D'$. However, this idea suffers from several drawbacks. First, the ellipsoid method is widely known to be  inefficient in practice. Second, the dual of $\D'$, denoted as $\P'$, is not our original LP $\P$ any more (e.g., $\P'$ now has exponentially many variables but polynomially many constraints) and it is not clear how it can help to recover the optimal solution to $\P$ instead. We believe that by opening the black box of Ellipsoid method and  after addressing the tedious arithmetic issues (since the separation oracle for $\D'$ needs to solve a convex program, which can only be solved to be within $\epsilon$ precision in $\poly(\log(1/\epsilon))$ time), one may indeed be able to recover the optimal solution to $\P$. However, this would take significantly more effort. Last but not least, the optimal solution obtained through such an optimization procedure is very unlikely to be a simple and interpretable mechanism.   

As a result, we choose to work with the dual of LP $\D'$, denoted as $\P'$, which recall that now has polynomially many constraints but exponentially many variables.  Surprisingly, we show that this transformation yields a simple characterization of the optimal mechanism as well as a much more efficient algorithm for computing it. This is enabled by two intriguing properties of   $\P'$ which we prove. First, any feasible solution to LP $\P'$ corresponds to a particular format of pricing outcomes mechanisms, which only either asks the buyer for a net payment of $b$ or pays the buyer the amount of $M$, and has no other format of transfers. Second, we prove that there always exists an optimal solution to LP $\P'$, which is a  $\texttt{CM-probR}$ mechanism, as described in Definition \ref{def:BDR}. As a result, we can conclude that $\texttt{CM-probR}$ can achieve the optimal revenue since the optimal objective of $\P'$ must equal that of $\P$ due to strong duality.

\subsection{Computing the Optimal $\texttt{CM-probR}$}

We now show how to efficiently compute the optimal $\texttt{CM-probR}$ mechanism via a polynomial-size linear program.  
For notational convenience, instead of $p_{\theta, b}(\omega, [a,+])$ and $p_{\theta, b}(\omega, [a,-])$ as in Definition \ref{def:BDR}, we use  variable $p^+_{\theta,b}(\omega, a)$ and $p^-_{\theta, b}(\omega, a)$ to denote the transition probabilities.  
Then the utility of the buyer with type $\theta$ and budget $b$ when he reports $(\theta', b')$ is equal to
$$
U_{\theta, b} (\theta', b')  =  \sum_{a} \Big( \max_{a'} \sum_\omega \mu(\omega | \theta, b) p^+_{\theta', b'}(\omega, a) (u( \omega, \theta, a') - b) + \max_{a'} \sum_\omega \mu(\omega | \theta, b) p^-_{\theta', b'}(\omega, a) (u( \omega, \theta, a') + M) \Big)
$$
Introduce new variables
\begin{eqnarray*}
&z^+_{\theta, b, \theta', b', a} =  \max_{a'} \sum_\omega \mu(\omega | \theta, b) p^+_{\theta', b'}(\omega, a) (u( \omega, \theta, a') - t^+), \\
&z^-_{\theta, b, \theta', b', a} =  \max_{a'} \sum_\omega \mu(\omega | \theta, b) p^-_{\theta', b'}(\omega, a) (u( \omega, \theta, a') - t^-),
\end{eqnarray*}
where $t^+ = b$ and $t^- = -M$. 
Then the LP for computing the optimal $\texttt{CM-probR}$ mechanism can be formulated as

\begin{lp} \label{prog:CMPR}
    \maxi{   \sum_{\omega, \theta, b} \mu(\omega, \theta, b) \sum_a \Big( b \cdot p^+_{\theta, b}(\omega, a) - M \cdot p^-_{\theta, b}(\omega, a) \Big)  }  
   \st 
    \qcon{  U_{\theta, b}( \theta, b) \ge U_{\theta, b} (\theta', b')}{ \theta, \theta', b \ge b'}{\text{(IC)}}
    \qcon{ U_{\theta, b}(\theta, b)  \ge \sum_\omega \mu(\omega | \theta, b) u(\omega, \theta, a')}{  (\theta, b), \  a' \in A}{\text{(IR)}}
    \qcon{  U_{\theta, b}(\theta, b) \le  \sum_{a, \omega, \circ\in\{+,-\}}  \mu(\omega | \theta, b) p^\circ_{\theta, b}(\omega, a) (u( \omega, \theta, a) - t^\circ)  }{ \theta, b}{\text{(OB)}}
    \qcon{ U_{\theta, b}(\theta', b') = \sum_a (z^+_{\theta, b, \theta', b', a} + z^-_{\theta, b, \theta', b', a})}{ \theta, b, \theta', b'}{}
    \qcon{   z^\circ_{\theta, b, \theta', b', a} \ge \sum_\omega \mu(\omega | \theta, b) p^\circ_{\theta', b'}(\omega, a) (u( \omega, \theta, a') - t^\circ)}{  \theta, b, \theta', b', a, a', \circ\in\{+,-\} }{}  
    \qcon{  \sum_{a \in A} [p^+_{\theta, b}(\omega, a) +p^-_{\theta, b}(\omega, a)]= 1}{  (\theta, b), \omega }{} 
    \qcon{0 \le p^+_{\theta, b}(\omega, a),p^+_{\theta, b}(\omega, a) \le 1}{ \theta, b, \omega, a}{}
\end{lp} 

These overall establish the following theorem.
\begin{theorem}
The optimal $\texttt{CM-probR}$ mechanism can be computed by a single linear program of $\poly(|\Theta|,|A|,|\Omega|)$ size for both public and private buyer budget setting. 
\end{theorem}

\section{Extension to Black-box Prior Distributions}\label{sec:black-box}
In this section, we extend our results to the setting where the seller  does \emph{not} know the prior distribution $\mu(\omega, \theta, b)$ and can only access the distribution by drawing i.i.d. samples from $\mu(\omega, \theta, b)$. We call this model the \emph{black-box prior distribution}.  This applies to the cases where distributions are only accessible through samples. We allow the signals to be correlated in this section, which includes independent signals as a special case. We provide a Monte-Carlo Sampling mechanism which draws polynomially many samples at the beginning and then computes an approximately optimal mechanism using the empirical distribution. This mechanism is nearly optimal and  approximately achieves incentive compatibility and individual rationality when the buyer knows that the samples are drawn from the true distribution. 
This mechanism is similar to the one in~\cite{dughmi2016algorithmic} which computes nearly optimal persuasion under unknown prior distributions. 
Our problem is more intricate because in persuasion the receiver's type (corresponding to the seller in our case) is known to the sender, but in our problem the buyer's type is randomly drawn from a distribution and needs to be sampled. 
For completeness, we fully state our algorithm and results here and defer the full proof to Appendix~\ref{sec:unknown}.

\begin{algorithm}[!h]               
\SetAlgorithmName{Mechanism}{mechanism}{List of Mechanisms}
\begin{algorithmic}
    \STATE (1) Observe private signal $\omega_1$. Ask the buyer to report his type $\theta_1$ and deposit his budget $b_1$.
    \STATE (2) Draw $n-1$ samples $T=\{(\theta_2, \omega_2, b_2), \dots, (\theta_n, \omega_n, b_n)\}$ from $\mu(\theta, \omega, b)$ to get 
$S=\{(\theta_1, \omega_1, b_1), \dots, $\\ $(\theta_n, \omega_n, b_n)\}\sim \mu^n$.
    \STATE (3) Solve an LP to get a $\texttt{CM-probR}$ mechanism $\mathcal{M}_S$. The LP will be defined later. 
    \STATE (4) Use $\mathcal{M}_S$ to sell the information to the buyer. Specifically, according to $\omega_1, \theta_1, b_1$, the mechanism \\ will  randomly decide an action to recommend, denoted by $a_S(\theta_1, b_1)$, together with a payment $t_S(\theta_1, b_1) \in\{ b_1, -M\}$.
\end{algorithmic}
\caption{$\texttt{CM-probR}$ Mechanism for Black-box Prior Distribution}
\label{alg:unknown}
\end{algorithm}

The seller will post the mechanism before observing $\omega_1$ and commit to follow the mechanism. The buyer does not know $\mu(\omega, \theta, b)$ but knows how the mechanism works. Our mechanism will satisfy the following properties.

\begin{definition}[$\varepsilon$-IC]
	Mechanism~\ref{alg:unknown} is $\varepsilon$-IC if for any possible $\mu(\omega, \theta, b)$, it is $\varepsilon$-optimal for the buyer to truthfully report and follow the mechanism's recommendation in expectation. More specifically, for any $\mu(\omega, \theta, b)$, any $(\theta, b) \neq (\theta', b')$ with $b'\le b$,
	\begin{eqnarray}
		&&\E_{\omega_1\sim \mu(\omega|\theta, b)} \ \E_{T\sim\mu^{n-1}} \  \E_{\mathcal{M}_S} \left[u(\omega_1, \theta, a_S(\theta, b)) -  t_S(\theta, b) \right] 
		\notag \\
		&\ge & \E_{\omega_1\sim \mu(\omega|\theta, b)} \ \E_{T\sim\mu^{n-1}} \  \E_{\mathcal{M}_S} \left[u(\omega_1, \theta, a'(a_S(\theta', b'))) - t_S(\theta', b')\right]- \varepsilon. \notag
	\end{eqnarray}
	for any deviation from the recommendation $a': A\to A$. 
\end{definition}

\begin{definition}[$\varepsilon$-IR]
Mechanism~\ref{alg:unknown} is $\varepsilon$-IR if for any possible $\mu(\omega, \theta, b)$, the buyer's expected surplus is no less than $-\varepsilon$ when he truthfully reports $(\theta, b)$ and follows the recommendation. More specifically, for any $\mu(\omega, \theta, b)$ and any $(\theta, b)$, 
$$
\E_{\omega_1\sim \mu(\omega|\theta, b)} \ \E_{T\sim \mu^{n-1}} \  \E_{\mathcal{M}_S} \left[u(\omega_1, \theta, a_S(\theta, b)) -  t_S(\theta, b) \right] \ge \E_{\omega_1\sim \mu(\omega|\theta, b)} \left[ u( \omega_1, \theta, a')\right]  -\varepsilon
$$
for all $a'\in A$. 
\end{definition}

\begin{definition}[$\varepsilon$-obedience]
Mechanism~\ref{alg:unknown} is $\varepsilon$-obedient if for any possible $\mu(\omega, \theta, b)$, when the buyer truthfully reports, in expectation it is $\varepsilon$-optimal for the buyer to take action $a$ when being recommended action $a$. More specifically,  for any $\mu(\omega, \theta, b)$, any $(\theta, b)$ and mapping $a': A \to A$, 
\begin{eqnarray}
&&\E_{\omega_1\sim \mu(\omega|\theta, b)} \ \E_{T\sim\mu^{n-1}} \  \E_{\mathcal{M}_S} \left[u(\omega_1, \theta, a_S(\theta, b))  \right] \notag\\
&\ge & \E_{\omega_1\sim \mu(\omega|\theta, b)} \ \E_{T\sim\mu^{n-1}} \  \E_{\mathcal{M}_S}  \left[u(\omega_1, \theta, a'(a_S(\theta, b)) )\right]  - \varepsilon. \notag
\end{eqnarray}
\end{definition}


We now define the LP that solves $\mathcal{M}_S$. The LP is basically the same as the one we use when $\mu$ is known~\eqref{prog:CMPR} but with estimated $\mu(\omega, \theta, b)$ and $\mu(\omega | \theta, b)$. For $\mu(\omega, \theta, b)$ in the objective function, we estimate it with the empirical distribution  over set $S$, denoted by $\widehat{\mu}(\omega, \theta, b)$. For $\mu(\omega | \theta, b)$ in the constraints, we estimate it with the empirical distribution over the set $S_{\theta, b} = \{\omega_1\} \cup \{\omega_i : i\ge 2 \text{ and } (\theta_i, b_i) = (\theta, b)\}$, denoted by $\widehat{\mu}(\omega| \theta, b)$. Note that the definition of $\widehat{\mu}(\omega| \theta, b)$ does \textbf{not} use the buyer's report $(\theta_1, b_1)$. The LP with variables $p^+, p^-$  is as follows.
\begin{lp} \label{prog:unknown}
    \maxi{ \sum_{\omega, \theta, b} \widehat{\mu}(\omega, \theta, b) \sum_a \Big( b \cdot p^+_{\theta, b}(\omega, a) - M \cdot p^-_{\theta, b}(\omega, a) \Big) }
    \st
    \qcon{U_{\theta, b}( \theta, b) \ge U_{\theta, b} (\theta', b') -\varepsilon}{\theta, \theta', b \ge b' }{\text{($\varepsilon$-IC)}}\\
    \qcon{U_{\theta, b}(\theta, b)  \ge \sum_\omega \widehat{\mu}(\omega | \theta, b) u(\omega, \theta, a') -\varepsilon}{ (\theta, b), \  a' \in A }{\text{($\varepsilon$-IR)}}\\
    \qcon{  U_{\theta, b}(\theta, b) - \varepsilon \le  \sum_{a, \omega, \circ\in\{+,-\}}  \widehat{\mu}(\omega | \theta, b) p^\circ_{\theta, b}(\omega, a) (u( \omega, \theta, a) - t^\circ)  }{ \theta, b}{\text{($\varepsilon$-OB)}}
    \qcon{ U_{\theta, b}(\theta', b') = \sum_a (z^+_{\theta, b, \theta', b', a} + z^-_{\theta, b, \theta', b', a})}{ \theta, b, \theta', b'}{}
    \qcon{   z^\circ_{\theta, b, \theta', b', a} \ge \sum_\omega \widehat{\mu}(\omega | \theta, b) p^\circ_{\theta', b'}(\omega, a) (u( \omega, \theta, a') - t^\circ)}{  \theta, b, \theta', b', a, a', \circ}{}  
    \qcon{  \sum_{a \in A} [p^+_{\theta, b}(\omega, a) +p^-_{\theta, b}(\omega, a)]= 1}{  (\theta, b), \omega }{} 
    \qcon{0 \le p^+_{\theta, b}(\omega, a),p^+_{\theta, b}(\omega, a) \le 1}{ \theta, b, \omega, a}{}
\end{lp} 
\begin{lemma} \label{lem:unknown}
When $\mathcal{M}_S$ is solved by the LP~\eqref{prog:unknown},
 Mechanism~\ref{alg:unknown} is $\varepsilon$-IR, $\varepsilon$-IC and 
 $\varepsilon$-obedient.
\end{lemma}
The lemma is proved in Appendix~\ref{app:unknown_lem}.
We then show that the mechanism extracts nearly optimal revenue with a sufficient number of samples. WLOG we assume $u \in [0,1]$, $M \in [0,1]$ and $b\in[0,1]$ for all $b\in B$.
We assume that $\mu_{\min} = \min_{\theta, b} \mu(\theta, b)$ is an instance-dependent constant.
\begin{theorem} \label{thm:unknown}
When we use Mechanism~\ref{alg:unknown} with LP~\eqref{prog:unknown} with 
$$
n \ge  \Theta\left( \ln(G/\delta) \cdot \max\left\{ \frac{|A|^2 }{\varepsilon^2 \cdot \mu_{\min}}, \frac{1}{\mu_{\min}^2} \right\}\right) = \Theta\left( \frac{|A|^2 \cdot \ln(G/\delta)}{\varepsilon^2 }\right),
$$ where $G = \max \{ |\Theta|, |B|, |A|\}$ and $\mu_{\min} = \min_{\theta, b} \mu(\theta, b)$, the mechanism will be $\varepsilon$-IR, $\varepsilon$-IC and 
$\varepsilon$-obedient, and extract revenue no less than $\texttt{Rev}_\mu(p^*) - \delta$.  $\texttt{Rev}_\mu(p^*)$ is the expected revenue of the optimal solution of~\eqref{prog:CMPR} $p^*$.
\end{theorem}
The theorem is proved in Appendix~\ref{app:unknown_thm}.

\section*{Acknowledgement}
We thank anonymous reviewers for the helpful comments.
\newpage
\bibliographystyle{ACM-Reference-Format}
\bibliography{ref}

\newpage
\appendix
%
%

\section{Omitted Proofs in Section \ref{sec:indep}}
\subsection{Proof of Theorem~\ref{thm:consult-public-opt}} \label{sec:public}

Our proof has two main steps. First, we show that it suffices to consider one-round mechanisms. We employ a characterization of  \cite{Babaioff:2012:OMS:2229012.2229024} who show that there always exists an optimal one-round mechanism which is a pricing mappings mechanism (defined formally next) for a buyer \emph{without} budget constraints. We generalize this characterization to the public budget case. Our second step is then to show that such a pricing mappings mechanism can always be converted to a consulting mechanism without revenue loss. 

\begin{definition}[Pricing mappings mechanism with direct payment]
A \emph{pricing mappings mechanism with direct payment} proceeds as follows:
	\begin{enumerate}
		\item Ask the buyer to report his type $\hat{\theta}$. His budget is publicly known as $b$.
		\item According to the reported type $\hat{\theta}$, charge the buyer an amount of $t_{\hat{\theta}} \le b$.
		\item Choose a signaling scheme $S(\hat{\theta})$ to send a signal to the buyer.
	\end{enumerate}  
	Formally, the protocol tree of a pricing mappings mechanism has three layers: a buyer node root, $|\Theta|$ transfer nodes below the root, and each transfer node connects to a signaling scheme $S(\theta)$.
\end{definition}

 We first show that there always exists an optimal mechanism which is a pricing mappings mechanism with direct payment. 
 
\begin{lemma}\label{lem:pricing-public-opt}
	When  the buyer has a public fixed budget $b$, and the buyer signal $\theta$ is independent from the seller signal $\omega$, there always exists an optimal mechanism which is an IC  pricing mappings mechanism with direct payment. 
\end{lemma}

\begin{proof}
	According to Lemma~\ref{lem:revelation}, there exists an optimal IC revelation mechanism $M$. We first reduce $M$ so that the buyer will always follow the protocol until a leaf is reached. If a type-$\theta$ buyer abort the protocol at some node $n$ after truthfully reporting his type, we can just cut away the subtree under $n$. We claim that this will not change the truthfulness guarantee and the revenue will also remain the same. First of all, for a buyer with type $\theta$, the protocol is the same because we only remove the part he is going to abort.  For buyers with other type $\theta' \neq \theta$,  if it is better for them to misreport $\theta$ when the subtree under $n$ is cut away, then it should also be better for them to misreport $\theta$ and abort at $n$ in the original protocol, which is contradictory to that $M$ is IC. Consequently, the mechanism is still truthful and everyone will pay the same amount of money. 
	
	Now let's just assume the buyer will always follow the protocol until reaching a leaf, after truthfully reporting his type. We construct a pricing mappings mechanism with direct payment that extracts the same revenue as $M$. Let $l_1, \dots, l_m$ be all the leaves of $M$ that can be reached with positive probability. Let $p_\theta(\omega, l_i)$ be the probability that the protocol ends at leaf $l_i$ when the seller's private signal is realized to $\omega$ and the buyer reports type $\theta$ at the beginning. Define $\tau(l_i)$ to be the total transfer on the path from the root to $l_i$. Since the signals are independent, the buyer's prior for $\omega$ is always the marginal distribution $\mu(\omega)$. As a result, for a buyer with any type, if he reports type $\theta$ and  follows the protocol until a leaf is reached, his expected total payment will always be equal to $\sum_{l_i} \tau(l_i) \sum_\omega \mu(\omega) p_{\theta}(\omega, l_i)$. Note that the budget $b$ should always allow the buyer to follow any path from the root to a leaf, in other words, the buyer will never have to quit the protocol  because he cannot afford the payment at some step. Because if this is not true, a buyer with a certain type will have to quit before reaching a leaf, which is contradictory to our assumption.
	
	Therefore we construct an optimal pricing mappings mechanism with direct payment as follows: 
	\begin{itemize}
	\item $t_{\theta} = \sum_{l_i} \tau(l_i) \sum_\omega \mu(\omega) p_{\theta}(\omega, l_i)$,
	\item the signaling scheme $S(\theta)$ has leaves $l_1, \dots, l_m$ and the transition function is just $p_{\theta}(\omega, l_i)$, so that the buyer's posterior at each $l_i$ remains the same.
	\end{itemize}
	This mechanism is still truthful because a buyer who gains from misreporting in this mechanism will prefer to misreport in the same way and then follow the protocol until reaching a leaf in $M$.  It should also hold that $t_{\theta} \le b$, since  $\tau(l_i) \le b$ for all $l_i$.  It is easy to see that the revenue remains the same.
	
	We emphasize that the above proof relies crucially on the assumptions that (1) the signals are independent and (2) there is a public fixed budget. If $\theta$ and $\omega$ are correlated, the buyer's prior for $\omega$ will be dependent on $\theta$, which is equal to $\mu(\omega | \theta)$.  This means that buyers with different types will have different expectations of how much they're going to pay. As a result, we cannot define $t_\theta$ using the common prior.   If the budget is not fixed, it may not be true that the buyer will always have enough budget to follow $M$ until reaching a leaf when he \emph{misreports}. So in the original mechanism $M$, the buyer may not be able to get the same utility as in the new mechanism when he misreports a type.  Consequently, the new mechanism may not be IC.
\end{proof}

Although we have now reduced the problem to designing simple one-round mechanisms,  the optimal signal schemes $S(\theta)$ can still be very complicated. 
Next, we show that any pricing mappings mechanism can be converted to a consulting mechanism without revenue loss. In other words, the signal sent by the seller can be converted to  a recommendation about the best action $a$. Since there are $|A|$ possible actions, the signal schemes $S(\theta)$ will have no more than $|A|$ leaves. Similar results appeared in Proposition 1 of~\cite{Bergemann2018Info}.

\begin{lemma} \label{lem:consult-public-opt}
For any IC \emph{pricing mappings mechanism with direct payment}, there exists an IC \emph{consulting mechanism with direct payment} that extracts no less revenue.
\end{lemma}

\begin{proof}
Let $M$ be any IC pricing mappings mechanism with direct payment.  Let $S(\theta)$ be the signaling scheme for the type-$\theta$ buyer.  Let $s_1, \dots, s_m$ be the leaves of $S(\theta)$, and let $a_1, \dots, a_m$ be the optimal actions of a type-$\theta$ buyer  when the signal sent by $S(\theta)$ is realized to $s_1, \dots, s_m$ respectively.  

Suppose $a_j, a_k$ both equal $\bar{a} \in A$ for some $j$ and $k$, we show that we can without loss merge the two leaves $s_j,s_k$.  That is,  whenever the seller moves to $s_j$, she moves to $s_k$ instead.  The transition function of $S(\theta)$ is updated as follows: 
	$p'(\omega, s_k) \leftarrow  p(\omega, s_j) + p(\omega, s_k)$ and  $p'(\omega, s_j)  \leftarrow 0$ for all $\omega$. The payments keep unchanged. Therefore, the seller's revenue will not change if the buyer still truthfully reports.  
	
	We claim that this new mechanism will not change the utility of a type-$\theta$ buyer. First of all, the optimal action for the type-$\theta$ buyer at leaf $s_k$ is still $\bar{a}$. This is because, by definition, we have  $$\bar{a} = \argmax_{a \in A} \sum_{\omega} \mu(\omega) p(\omega, s_j) u(\omega, \theta, a), $$ $$\bar{a} = \argmax_{a \in A} \sum_{\omega} \mu(\omega) p(\omega, s_k) u(\omega, \theta, a).$$  
	Therefore $\argmax_{a \in A} \sum_{\omega} \mu(\omega) p'(\omega, s_j) u({\omega}, \theta, a)$ must also equal $\bar{a}$. As a result, the expected utility of type-$\theta$ buyer  will not change as $ \sum_{\omega} \mu(\omega) p(\omega, s_j) u(w, \theta, \bar{a}) +  \sum_{w} \mu(\omega) p(\omega, s_k) u({\omega}, \theta, \bar{a}) = \sum_{w} \mu(\omega) p'(\omega, s_j) u(w, \theta, \bar{a}) $. 
	
	Next we argue that  the utility of any other type-$\theta'$ buyer will not increase when he misreport  $\theta$.  This is because his utility from the original two leaves $s_j ,s_k$ is 
 $$\max_{a \in A} \sum_{\omega} \mu(\omega) p(\omega, s_j)  u({\omega}, \theta', a) $$ 
 and 
 $$\max_{a \in A} \sum_{\omega} \mu(\omega) p(\omega, s_k) u({\omega}, \theta', a) .$$
 Their sum is at least 
 $$\max_{a \in A} \sum_{\omega} \mu(\omega) [ p(\omega, s_j) + p(\omega, s_k)] \cdot u(w, \theta', a) .
 $$ 
 
To sum up, the new mechanism will not change the utility of type-$\theta$ buyer  and will not increase the other buyers' utility when they misreport $\theta$, and thus will remain IC and IR. Moreover, the revenue does not change. We can perform such merging operation until each leaf corresponds to a different action $a \in A$. To that end, each signal can be viewed as an ``honest'' action recommendation which will indeed maximizes the buyer's expected utility. This is precisely a consulting mechanism, as desired. 

\end{proof}

Lemma \ref{lem:pricing-public-opt} and \ref{lem:consult-public-opt} together yield a proof of Theorem~\ref{thm:consult-public-opt}.

\subsection{Proof of Theorem~\ref{thm:consult-private-optimal}} \label{sec:private}
The proof consists of two parts as the public budget case. We will first prove the existence of a one-round optimal mechanism, a \emph{pricing mappings mechanism with deposit and return}, and then convert this mechanism to a  \emph{consulting mechanism with deposit and return}.

\begin{definition}[Pricing mappings mechanism with deposit and return]
A \emph{pricing mappings mechanism with deposit and return} proceeds as follows:
	\begin{enumerate}
		\item Ask the buyer to report his type $\hat{\theta}$ and budget $\hat{b}$, and then deposit his budget $\hat{b}$.
		\item According to the reported type $\hat{\theta}$ and the deposit amount $\hat{b}$, return the buyer an amount of $\hat{b} - t_{\hat{\theta}, \hat{b}} \ge 0$.
		\item Choose a signaling scheme $S(\hat{\theta}, \hat{b})$ to send a signal to the buyer.
	\end{enumerate}  
	Formally, the protocol tree of a pricing mappings mechanism has three layers: a buyer node root, $|\Theta| \cdot |B|$ transfer nodes below the root, and each transfer node connects to a signaling scheme $S(\theta, b)$.
\end{definition}

\begin{lemma}\label{lem:pricing-private-opt}
	When the buyer has a private budget, and $(\theta, b)$ is independent from $\omega$,  there always exists an optimal mechanism which is an IC  pricing mappings mechanism with deposit and return. 
\end{lemma}

\begin{proof}
	According to Lemma~\ref{lem:revelation}, there exists an optimal IC revelation mechanism $M$.   As in the proof of Lemma~\ref{lem:pricing-public-opt}, we can reduce $M$ so that the buyer will always follow the protocol until reaching a leaf.
Now let's just assume the buyer will always follow the protocol until reaching a leaf, after truthfully reporting his type and budget. We construct a pricing mappings mechanism with deposit and return that extracts the same revenue as $M$. Let $l_1, \dots, l_m$ be all the leaves of $M$ that can be reached with positive probability. Let $p_{\theta,b}(\omega, l_i)$ be the probability that the protocol ends at leaf $l_i$ when the seller's private signal is realized to $\omega$ and the buyer reports $\theta$ and $b$ at the beginning. Define $\tau(l_i)$ to be the total transfer on the path from the root to $l_i$. Since  $(\theta, b)$ is independent from $\omega$, the buyer's prior for $\omega$ is always the marginal distribution $\mu(\omega)$. As a result, for a buyer with any type and budget, if he reports type $\theta, b$ and  follows the protocol until a leaf is reached, his expected total payment will always be equal to $\sum_{l_i} \tau(l_i) \sum_\omega \mu(\omega) p_{\theta,b}(\omega, l_i)$.	
	Therefore we construct an optimal pricing mappings mechanism with deposit and return as follows: 
	\begin{itemize}
	\item $t_{\theta, b} = \sum_{l_i} \tau(l_i) \sum_\omega \mu(\omega) p_{\theta,b}(\omega, l_i)$,
	\item the signaling scheme $S(\theta, b)$ has leaves $l_1, \dots, l_m$ and the transition function is just $p_{\theta, b}(\omega, l_i)$, so that the buyer's posterior at each $l_i$ remains the same.
	\end{itemize}
	Note that the buyer can never overreport his budget in this mechanism because he is asked to deposit the full budget at the beginning. This guarantees the truthfulness of the new mechanism because a $(\theta, b)$-buyer who gains from misreporting $(\theta', b')$ with $b'\le b$ can also get higher utility in $M$ by misreporting $(\theta', b')$  and following the protocol until a leaf is reached. Since $b \ge b'$, he should always have enough budget to reach the leaves in $M$ if he reports $b'$. It should also hold that $t_{\theta, b} \le b$, since for any leaf $l_i$ that can be reached by $(\theta, b)$-buyer with positive probability must have $\tau(l_i) \le b$.  Finally it is easy to see that the revenue remains the same.
\end{proof}

Same as the public budget case, any \emph{pricing mappings mechanism with deposit and return} can be without loss converted into a mechanism that recommends actions to the buyer, with the same \emph{deposit and return} payment method.

\begin{lemma} \label{lem:consult-private-opt}
For any IC \emph{pricing mappings mechanism with deposit and return}, there exists an IC \emph{consulting mechanism with deposit and return} that extracts no less revenue.
\end{lemma}

The proof follows by a simple adaptation of the proof of Lemma~\ref{thm:consult-public-opt}, with $S(\theta, b)$ in place of $S(\theta)$.

\section{Proof of Theorem \ref{thm:BDR-opt} }\label{app:corr_lem}

We start by describing another class of mechanisms defined as follows, which is a strictly broader class than $\texttt{CM-probR}$ mechanisms and has been proved to contain an optimal mechanism for the setting with no budget. 

\begin{definition}[Pricing outcomes mechanism with deposit and return, \cite{Babaioff:2012:OMS:2229012.2229024}]
	A \emph{pricing outcomes mechanism with deposit and return} (POM-depR) proceeds as follows:
	\begin{enumerate}
		\item The seller commits to a signaling scheme for each $(\theta, b)$, described by $\{ p_{\theta,b}(\omega, \sigma) \}_{\omega \in \Omega, \sigma \in \Sigma}$  where $p_{\theta,b}(\omega, \sigma)$ denotes the probability of sending signal $\sigma$ upon seller signal $\omega$, and a payment amount $t_{\theta,b}(\sigma) (\leq b)$ that depends on the realized signal $\sigma$.  
		\item Ask the buyer to report his type $\hat{\theta}$ and deposit his budget $\hat{b}$.
		\item Sample a signal $\sigma \sim p_{\hat{\theta},\hat{b}}(\omega)$ and sends $\sigma$ to the buyer
		
		\item Return $\hat{b} - t_{\hat{\theta}, \hat{b}}(\sigma)$ back to the buyer. 
	\end{enumerate}  
\end{definition}

As our first step, we generalize the result of \cite{Babaioff:2012:OMS:2229012.2229024} and  show that for budget-constrained buyers, pricing outcomes mechanisms with deposit and return (POM-depR) can also achieve optimality.  
\begin{lemma}\label{lem:signal-opt}
	For a budget-constrained buyer, there always exists an IC pricing outcomes mechanism with deposit and return (POM-depR) that maximizes the revenue.
\end{lemma}
\begin{proof}
	According to Lemma~\ref{lem:revelation}, there exists an IC optimal revelation mechanism $M$. 
	We can reduce $M$ so that the buyer will always follow the protocol until reaching a leaf.
	Now let's just assume the buyer will always follow the protocol until reaching a leaf, after truthfully reporting his type and budget. We construct a pricing outcomes mechanism with deposit and return that extracts the same revenue as $M$. Let $l_1, \dots, l_m$ be all the leaves of $M$ that can be reached with positive probability. Let $p_{\theta,b}(\omega, l_i)$ be the probability that the protocol ends at leaf $l_i$ when the seller's private signal is realized to $\omega$ and the buyer reports $\theta$ and $b$ at the beginning. Define $\tau(l_i)$ to be the total transfer on the path from the root to $l_i$. 
	Then we construct an optimal pricing outcomes mechanism as follows: 
	\begin{itemize}
		\item the signaling scheme $S(\theta, b)$ has leaves $l_1, \dots, l_m$ and the transition function is just $p_{\theta, b}(\omega, l_i)$, 
		\item the refund below leaf $l_i$ is set to $b - \tau(l_i)$ so that the net transfer equals $\tau(l_i)$. 
	\end{itemize}
	Since the buyer will never defect in $M$, the refund amount $b - \tau(l_i)$ is always nonnegative, so the buyer will not have chance to defect in the last step.
	
	This pricing outcomes mechanism is still truthful because a buyer who gains from misreporting in this mechanism will prefer to misreport in the same way in $M$. Then since the net transfer remains the same on each path, and the probability of ending at each leaf remains the same, the revenue does not change.  
\end{proof}

Observe that any $\texttt{CM-probR}$ mechanism  is also a POM-depR mechanism since we can view each $[a,i] \in A \times \{ +,- \}$ as a signal. However, $\texttt{CM-probR}$ is significantly simpler --- it only has two payments, i.e., $b$ and $-M$, and it uses at most $2n$ signals. Our goal is to show that this much smaller class of mechanisms can still achieve optimality.  For notational convenience, we will prove our result for the case of public budget. It is straightforward to generalize the argument to private budget, simply by adding additional incentive compatibility constraints to  elicit truthful budget report (or equivalently, by viewing the buyer type $\theta$ as including the budget information).

Our starting point is  an LP formulation  proposed by \cite{Babaioff:2012:OMS:2229012.2229024} with exponentially many variables but polynomially many constraints for computing the optimal POM-depR mechanism.  We now adapt that LP to incorporate buyer budget constraints, which leads to an LP with exponentially many variables \emph{and} constraints.   

To describe the LP, we will need a different but equivalent definition of a signaling scheme. Recall that we use $\{ p_{\theta}(\omega, \sigma) \}_{\omega \in \Omega,\sigma\in \Sigma}$ to denote a signaling scheme for buyer type $\theta$ (omitting the subscript of budget) where $\sum_{\sigma \in \Sigma} p_{\theta}(\omega, \sigma) = 1$ for each $\omega$,  which is drawn from the prior distribution  $\{ \mu(\omega) \}_{\omega \in \Omega}$. Equivalently, this signaling scheme can be viewed as a \emph{convex decomposition} of the prior $\mu \in \Delta \Omega$ into a set of posterior distributions. In particular, the posterior distribution $\{ \frac{ p_{\theta}(\omega, \sigma) }{  \sum_{\omega' \in \Omega} p_{\theta}(\omega', \sigma) } \}_{\omega \in \Omega}$ has convex coefficient $\sum_{\omega' \in \Omega} p_{\theta}(\omega', \sigma)$ (i.e., the probability).


As a result, a signaling scheme $p_{\theta}$ can be equivalently described by a variable $x_{\theta}(q) \in [0,1]$ for each possible posterior distribution $q \in \Delta \Omega$, satisfying $\sum_{q} x_{\theta}(q) \cdot q = \mu$ so that $\{x_{\theta}(q) \}_{q \in \Delta \Omega}$ indeed represents a convex decomposition of $\mu$. Crucially, given any posterior distribution $q$, each buyer will interpret it differently due to different information the buyer type $\theta$ possess regarding $\omega$. 
Concretely, as observed by \cite{Babaioff:2012:OMS:2229012.2229024}, buyer type $\theta$ will interpret any posterior $q$ as $D_{\theta} q$ (up to a normalization factor) where  $D_{\theta}$ is a $\Omega \times \Omega$ diagonal matrix with $\mu(\theta|\omega)$ for all $\omega $ in the diagonal. A minor issue is that we now have infinitely many variables since each $q \in \Delta \Omega$ corresponds to a variable $x_{\theta}(q)$. This turns out to be fine for our approach. However, for convenience, we invoke a simplification by \cite{Babaioff:2012:OMS:2229012.2229024} who show that one can without loss of generality focus on a set $Q^*$ of finite but \emph{exponentially} many possible posterior beliefs $q$'s since only these $q$'s can show up in the optimal solution.  

Now we are ready to state the linear program for computing the optimal POM-depR mechanism, with variable $x_{\theta}(q)$ and  $\tilde{t}_{\theta}(q) = x_{\theta}(q) \cdot t_{\theta}(q)$ where  $t_{\theta}(q)$ is the corresponding buyer total payment for the realized signal inducing posterior $q$ (negative $t_{\theta}(q)$ means that the seller will pay the buyer).  

\begin{lp}\label{lp:posterior-opt-budget}
	\maxi{  \sum_{\theta, q} (\bm{1}^t D_{\theta} q) \cdot \tilde{t}_{\theta}(q) }
	\st 
	\qcon{ \sum_q [ v_{\theta }(D_{\theta} q) x_{\theta}(q) - (\bm{1}^t D_{\theta} q) \tilde{t}_{\theta}(q)]  \ge v_{\theta}(D_{\theta } \mu)}{ \theta \in \Theta}{} 
	\con{ \sum_q [ v_{\theta }(D_{\theta } q) x_{\theta}(q) - (\bm{1}^t D_{\theta } q) \tilde{t}_{\theta }(q) ]  }{}
	\qcon{ \qquad \qquad  \ge  \sum_q [ v_{\theta }(D_{\theta } q) x_{\theta'}(q) - (\bm{1}^t D_{\theta } q) \tilde{t}_{\theta' }(q)] }{ \theta'\neq \theta}{}
	\qcon{  \sum_q x_{\theta}(q)  \cdot q = \mu  }{ \theta \in \Theta}{}
	\qcon{  b \cdot x_{\theta}(q) \geq \tilde{t}_{\theta}(q) \ge -M \cdot x_{\theta}(q) }{ \theta \in \Theta, q \in Q^*}{}
	\qcon{  x_{\theta}(q) \ge 0}{\theta \in \Theta, q \in Q^*}{}
\end{lp}
where $v_{\theta}(q) = \max_{a' \in A} \Ex_{\omega \sim q} u(\omega, \theta, a')$ denotes the type-$\theta$ buyer's optimal expected utility when he believes that $\omega$ follows distribution $q$.  Note that $q$ here does not have to be normalized to have $l_1$ norm $1$ since $v_{\theta}(q)$ is linear in any re-scaling factor of $q$. 

We mention a subtle issue here regarding the above LP formulation. Our LP \eqref{lp:posterior-opt-budget} is similar to an LP formulation given in \cite{Babaioff:2012:OMS:2229012.2229024} for computing the optimal pricing outcomes mechanism \emph{without} budget constraints. More concretely, the LP of \cite{Babaioff:2012:OMS:2229012.2229024} does not have the constraint $b \cdot x_{\theta}(q) \geq \tilde{t}_{\theta}(q) \ge -M \cdot x_{\theta}(q) $.  However, that LP does \emph{not} exactly compute the optimal pricing outcomes mechanism since it may have an optimal solution such that $\tilde{t}_{\theta}(q) \not = 0$ but $x_{\theta}(q) = 0$ for some $\theta, q$, which does not correspond to any pricing outcomes mechanisms. This discrepancy is due to their variable exchange $\tilde{t}_{\theta}(q) =x_{\theta}(q) t_{\theta}(q) $ which  does not result in an equivalent formulation since the constraint $x_{\theta}(q) = 0 \Rightarrow \tilde{t}_{\theta}(q) = 0$ cannot not be enforced after the variable change (and cannot be casted as any linear constraint). To deal with this issue, \cite{Babaioff:2012:OMS:2229012.2229024} introduced a technique of slightly perturbing the optimal solution and look for an $\epsilon$-optimal solution which always satisfies $x_{\theta}(q) \not = 0$ whenever $\tilde{t}_{\theta}(q) \not = 0$. As a result, this technique can only give an approximately optimal pricing outcomes mechanism, and may result in infinitely large payment as the approximation error $\epsilon$ tends to $0$. 

Interestingly, it turns out that the budget constraints help us to get rid of the above discrepancy.  In particular, our LP \eqref{lp:posterior-opt-budget} exactly computes the optimal pricing outcomes mechanism. This is because the budget constraint  $b \cdot x_{\theta}(q) \geq \tilde{t}_{\theta}(q) \ge -M \cdot x_{\theta}(q) $ automatically implies that $\tilde{t}_{\theta}(q) = 0$ when $x_{\theta}(q) = 0$, assuming all budgets are finite. We summarize this observation in the following lemma.   
\begin{lemma}
	The optimal solution to LP \eqref{lp:posterior-opt-budget} corresponds exactly to an optimal pricing outcomes mechanism.  
\end{lemma}

The main challenge for solving LP \eqref{lp:posterior-opt-budget} is that it has both exponentially many variables and constraints (due to budget constraints) since the set $Q^*$ is exponentially large. It is not even clear that it has an optimal solution with polynomial size, let alone solving it to output the optimal solution in polynomial time.  

Our only bet is that the exponential-sized budget constraints appear relatively ``simple'', and  hopefully it will result in dual variables that are easy to handle. With this in mind, we consider the dual program of LP \eqref{lp:posterior-opt-budget}, formulated as follows where $ \{ \alpha_{\theta} \}_{\theta \in \Theta},  \{ \lambda_{\theta, \theta'} \}_{\theta \not = \theta'}, \{ y_{\theta} \in \RR^{\Omega} \}_{\theta \in \Theta}$ and $\{\gamma_{\theta, q}, \xi_{\theta, q} \}_{\theta, q}$ are the dual variables for the first, second, third and forth sets of constraints in LP \eqref{lp:posterior-opt-budget}, respectively. 

\begin{lp}\label{lp:posterior-opt-budget-dual}
	\mini{  \sum_{\theta} \mu \cdot y_{\theta} - v_{\theta}(D_{\theta} \mu) \cdot \alpha_{\theta} }
	\st 
	\qcon{ \gamma_{\theta, q} - \xi_{\theta,q} - (\bm{1}^t D_{\theta} q) (1  - \alpha_{\theta}   - \sum_{\theta' \neq \theta} \lambda_{\theta, \theta'} )  -  \sum_{\theta' \neq \theta} \lambda_{\theta', \theta} (\bm{1}^t D_{\theta'} q)  = 0}{\theta, q}{}
	\qcon{y_{\theta} \cdot q - b \gamma_{\theta,q} - M \xi_{\theta, q} - v_{\theta }(D_{\theta} q) ( \alpha_{\theta} + \sum_{\theta' \neq \theta} \lambda_{\theta, \theta'})   + \sum_{\theta' \neq \theta} \lambda_{\theta', \theta} v_{\theta'}(D_{\theta'} q)  \geq 0}{\theta, q}{}
	\con{\alpha, \lambda, \gamma, \xi \geq 0}{}
\end{lp}

Obviously, LP \eqref{lp:posterior-opt-budget-dual} is not tractable neither as it still has exponentially many variables and constraints. Our next main step is to prove that the variables $\gamma_{\theta,q}, \xi_{\theta,q}$ in LP \eqref{lp:posterior-opt-budget-dual} can be eliminated to induce a different linear program which has the same optimal objective as LP \eqref{lp:posterior-opt-budget-dual} but only has polynomially many variables. For notational convenience, we define the following functions: 
\begin{eqnarray*}
	g_\theta(\lambda, \alpha; q) &=& (\bm{1}^t D_{\theta} q) (1  - \alpha_{\theta}   - \sum_{\theta' \neq \theta} \lambda_{\theta, \theta'} )  + \sum_{\theta' \neq \theta} \lambda_{\theta', \theta} (\bm{1}^t D_{\theta'} q)\\
	h_\theta(\lambda, \alpha; q) &=& v_{\theta }(D_{\theta} q) ( \alpha_{\theta} + \sum_{\theta' \neq \theta} \lambda_{\theta, \theta'})   - \sum_{\theta' \neq \theta} \lambda_{\theta', \theta} v_{\theta'}(D_{\theta'} q).
\end{eqnarray*}

Note that $g_{\theta}, h_{\theta}$ are both \emph{linear} functions of $\alpha, \lambda$ for any given $q \in Q^*$. Using the above simplification, we can re-write LP \eqref{lp:posterior-opt-budget-dual} as follows.   

\begin{figure}[ht] 
	\begin{lp}\label{lp:posterior-opt-budget-dual-eqv}
		\mini{  \sum_{\theta} \mu \cdot y_{\theta} - v_{\theta}(D_{\theta} \mu) \cdot \alpha_{\theta} }
		\st 
		\qcon{ \gamma_{\theta, q} - \xi_{\theta,q} - g_\theta(\lambda, \alpha; q)  = 0}{\theta, q}{}
		\qcon{y_{\theta} \cdot q - b \gamma_{\theta,q} - M \xi_{\theta, q} - h_\theta(\lambda, \alpha; q)  \geq 0}{\theta, q}{}
		\con{\alpha, \lambda, \gamma, \xi \geq 0}{}
	\end{lp}
	\vspace{-3mm}
	\caption{Re-Writing of LP   \eqref{lp:posterior-opt-budget-dual}    }
\end{figure}

\begin{lemma}
	LP \eqref{lp:posterior-opt-budget-dual-eqv} has the same optimal objective value as the following LP with only variables $\alpha, \lambda, y$: 
	\begin{lp}\label{lp:posterior-opt-budget-dual-var}
		\mini{  \sum_{\theta} \mu \cdot y_{\theta} - v_{\theta}(D_{\theta} \mu) \cdot \alpha_{\theta} }
		\st 
		\qcon{y_{\theta} \cdot q - b \cdot g_{\theta}  - h_\theta(\lambda, \alpha, q) \geq 0}{\theta, q}{}
		\qcon{y_{\theta} \cdot q + M  \cdot g_{\theta}  - h_\theta(\lambda, \alpha, q) \geq 0}{\theta, q}{}
		\con{\alpha, \lambda \geq 0}{}
	\end{lp}
\end{lemma} 
\begin{proof}
	The main idea of the proof is to eliminate the variables $\gamma_{\theta, q}, \xi_{\theta, q}$. 
	
	First, we claim that LP \eqref{lp:posterior-opt-budget-dual-eqv} always admits an optimal solution in which either $\gamma_{\theta,q} = 0$ or $\xi_{\theta,q} = 0$ for any $\theta, q$. In particular, let $(\alpha^*, \lambda^*, y^*, \gamma^*, \xi^*)$ be any optimal solution to LP \eqref{lp:posterior-opt-budget-dual-eqv}. We construct a new optimal solution which satisfies the desired property, as follows. For any index pair $(\theta, q)$: 
	\begin{itemize}
		\item  If $\gamma_{\theta,q}^* \geq \xi^*_{\theta,q}$, define $\gamma_{\theta, q}' = \gamma_{\theta,q}^* -  \xi^*_{\theta,q}$ and $\xi_{\theta,q}' = 0$;
		\item If $\gamma_{\theta,q}^* < \xi^*_{\theta,q}$, define $\gamma_{\theta, q}' = 0$ and $\xi_{\theta,q}' = \xi^*_{\theta,q}  -  \gamma_{\theta,q}^* $.
	\end{itemize}   
	
	We claim that $(\alpha^*, \lambda^*, y^*, \gamma', \xi')$ is a feasible and optimal solution to LP \eqref{lp:posterior-opt-budget-dual}. The new variable values achieve the same (optimal) objective value as $(\alpha^*, \lambda^*, y^*, \gamma^*, \xi^*)$ since the objective function only depends on $\alpha^*,y^*$. We only need to argue that $(\alpha^*, \lambda^*, y^*, \gamma', \xi')$  is feasible. The non-negativity of $\gamma'_{\theta,q}, \xi'_{\theta,q}$  is obvious from construction. The first set of constraint is satisfied  because $ \gamma_{\theta,q}^* -  \xi^*_{\theta,q} = \gamma'_{\theta,q} -  \xi'_{\theta,q} $ in our construction and $g_{\theta}$ does not depend on $\gamma_{\theta,q} , \xi_{\theta,q}$. The second set of constraints are also feasible because  $y^*_{\theta} \cdot q - b \gamma^*_{\theta,q} - M \xi^*_{\theta, q} - h_\theta(\lambda^*, \alpha^*; q) \geq 0$, and the  constructed variable value $\gamma'_{\theta,q} \xi'_{\theta, q} $ satisfy $\gamma'_{\theta,q} \leq  \gamma^*_{\theta,q}, \xi'_{\theta, q}  \leq   \xi^*_{\theta, q} $ and thus will only  increase the left-hand side of the above inequality. Therefore the ``$\geq$'' inequality still holds.  
	
	Thanks to the above argument, we can impose the complementarity constraints $ \gamma_{\theta, q} \cdot  \xi_{\theta,q}  = 0$ into LP \eqref{lp:posterior-opt-budget-dual-eqv}, without changing its optimal objective value.  As a result, the equality constraint  $\gamma_{\theta, q} - \xi_{\theta,q} - g_\theta(\lambda, \alpha; q)  = 0$ in LP \eqref{lp:posterior-opt-budget-dual-eqv}, together with $ \gamma_{\theta, q} \cdot  \xi_{\theta,q}  = 0$,  can be used to uniquely determine the value of $\gamma_{\theta, q}$ and $ \xi_{\theta,q}$. In particular, if $g_\theta(\lambda, \alpha; q) \geq 0$, we must have $\gamma_{\theta, q} = g_\theta(\lambda, \alpha; q) $ and $\xi_{\theta,q} = 0$. Otherwise, we must have $\gamma_{\theta, q} = 0 $ and $\xi_{\theta,q} = -g_\theta(\lambda, \alpha; q)$. To sum up, $\gamma_{\theta, q}$ and $ \xi_{\theta,q}$ can be uniquely determined as follows: 
	\begin{eqnarray*}
		\gamma_{\theta, q}  &=& \mathds{1}\left\{ g_\theta(\lambda, \alpha, q) \ge 0 \right\} g_\theta(\lambda, \alpha, q), \\
		\xi_{\theta, q} &=& - \mathds{1}\left\{ g_\theta(\lambda, \alpha,  q) \le 0 \right\} g_\theta(\lambda, \alpha, q),
	\end{eqnarray*}
	where $\mathds{1}(X)$ is the indicator function which equals $1$ when condition $X$ holds. 
	
	We can now plug in the above value of  $\gamma_{\theta, q}, \xi_{\theta, q} $ into the left-hand side of the second constraint of LP \eqref{lp:posterior-opt-budget-dual-eqv}, and obtain the following inequality constraint for any $\theta, q$: 
	\begin{eqnarray}\label{eq:constraint-rewrite}
	y_{\theta} \cdot q - b \cdot \mathds{1}\left\{ g_\theta(\lambda, \alpha, q) \ge 0 \right\} g_\theta(\lambda, \alpha, q)  + M \cdot  \mathds{1}\left\{ g_\theta(\lambda, \alpha,  q) \le 0 \right\} g_\theta(\lambda, \alpha, q)  - h_\theta(\lambda, \alpha; q) \geq 0
	\end{eqnarray}
	
	Observe that 
	\begin{eqnarray*}
		& &  - b \cdot \mathds{1}\left\{ g_\theta(\lambda, \alpha, q) \ge 0 \right\} g_\theta(\lambda, \alpha, q)  + M \cdot  \mathds{1}\left\{ g_\theta(\lambda, \alpha,  q) \le 0 \right\} g_\theta(\lambda, \alpha, q)  \\
		& = & \min \{ -b \cdot g_\theta(\lambda, \alpha, q), M \cdot  g_\theta(\lambda, \alpha, q)  \}  
	\end{eqnarray*} 
	which is the minimum of two linear functions and thus is concave. As a result, Inequality constraint \eqref{eq:constraint-rewrite} can be re-written as the following
	\begin{eqnarray*}
		& & y_{\theta} \cdot q - b \cdot \mathds{1}\left\{ g_\theta(\lambda, \alpha, q) \ge 0 \right\} g_\theta(\lambda, \alpha, q)  + M \cdot  \mathds{1}\left\{ g_\theta(\lambda, \alpha,  q) \le 0 \right\} g_\theta(\lambda, \alpha, q)  - h_\theta(\lambda, \alpha; q) \geq 0 \\
		&\Leftrightarrow &  y_{\theta} \cdot q + \min \{ -b \cdot g_\theta(\lambda, \alpha, q), M \cdot  g_\theta(\lambda, \alpha, q)  \}     - h_\theta(\lambda, \alpha; q) \geq 0  \\
		&\Leftrightarrow & \min \{ -b \cdot g_\theta(\lambda, \alpha, q), M \cdot  g_\theta(\lambda, \alpha, q)  \}  \geq    h_\theta(\lambda, \alpha; q) -y_{\theta} \cdot q    \\ 
		&\Leftrightarrow & -b \cdot g_\theta(\lambda, \alpha, q)   \geq    h_\theta(\lambda, \alpha; q) -y_{\theta} \cdot q  \quad \&   \quad  M \cdot  g_\theta(\lambda, \alpha, q)   \geq    h_\theta(\lambda, \alpha; q) -y_{\theta} \cdot q  
	\end{eqnarray*} 
	Substituting Inequality constraint \eqref{eq:constraint-rewrite} by  the above to inequalities, we obtain LP \eqref{lp:posterior-opt-budget-dual-var} which has the same optimal objective as LP \eqref{lp:posterior-opt-budget-dual-eqv}.  
\end{proof}

We now plugging the concrete form of $g$ and $h$ into LP \eqref{lp:posterior-opt-budget-dual-var} and obtain the following concrete linear program, which has the same objective as LP \eqref{lp:posterior-opt-budget-dual} but with only polynomially many constraints.
\begin{figure}[ht]
	\begin{lp}\label{lp:posterior-opt-budget-dual-variant}
		\mini{  \sum_{\theta} \mu \cdot y_{\theta} - v_{\theta}(D_{\theta} \mu) \cdot \alpha_{\theta} }
		\st 
		\con{ - \alpha_{\theta} \bigg[ v_{\theta} (D_{\theta} q) - b(\bm{1}^t D_{\theta } q)  \bigg] - \sum_{\theta'} \lambda_{\theta, \theta'} \bigg[ v_{\theta} (D_{\theta} q) - b(\bm{1}^t D_{\theta } q)  \bigg]  + }{} 
		\qcon{ \qquad  \sum_{\theta'} \lambda_{\theta', \theta}\bigg[ v_{\theta'} (D_{\theta'} q) - b(\bm{1}^t D_{\theta' } q)  \bigg]  + y_{\theta} q \geq b(\bm{1}^t D_{\theta } q)   }{\theta, q}{}
		\con{ - \alpha_{\theta} \bigg[ v_{\theta} (D_{\theta} q) + M(\bm{1}^t D_{\theta } q)  \bigg] - \sum_{\theta'} \lambda_{\theta, \theta'} \bigg[ v_{\theta} (D_{\theta} q)  + M(\bm{1}^t D_{\theta } q)  \bigg]  + }{}
		\qcon{ \qquad  \sum_{\theta'} \lambda_{\theta', \theta}\bigg[ v_{\theta'} (D_{\theta'} q)  + M (\bm{1}^t D_{\theta' } q)  \bigg]  + y_{\theta} q \geq -M (\bm{1}^t D_{\theta } q)   }{\theta, q}{}
		\con{\alpha, \lambda \geq 0}{}
	\end{lp}
	\vspace{-3mm}
	\caption{Re-Writing of LP \eqref{lp:posterior-opt-budget-dual-var}}
\end{figure}

At this point, a natural approach is to design an efficient separation oracle for LP \eqref{lp:posterior-opt-budget-dual-variant} and then employ the celebrated Ellipsoid method to solve it with a hope to also obtain an optimal solution to LP \eqref{lp:posterior-opt-budget}. Due to the reasons we mentioned in the main body, here we adopt a different approach and instead turn to the dual program of LP \eqref{lp:posterior-opt-budget-dual-variant}. However, the benefit of considering the dual program is not clear at this point yet since it will have exponentially many variables. Interestingly, we show that by exploring some ``local convexity'' structure of the formulation the dual program actually can be solved much more efficiently.   

We start by deriving the dual of LP \eqref{lp:posterior-opt-budget-dual-variant} as follows, with variable $x_{\theta}^+(q), x_{\theta}^{-}(q)$ for the two different sets of constraints. 

\begin{lp}\label{lp:primal-opt}
	\maxi{  \sum_{\theta, q} \bigg[  b(\bm{1}^t D_{\theta } q) \cdot x^+_{\theta}(q) - M(\bm{1}^t D_{\theta } q) \cdot x^-_{\theta}(q) \bigg]  }
	\st 
	\qcon{ \sum_q \bigg[ [v_{\theta }(D_{\theta} q) -b(\bm{1}^t D_{\theta } q)]  x_{\theta}^+(q) + [v_{\theta }(D_{\theta} q) +M(\bm{1}^t D_{\theta } q)]  x_{\theta}^-(q)  \bigg]  \ge v_{\theta}(D_{\theta } \mu)}{ \theta \in \Theta}{}
	\con{ \sum_q \bigg[ [v_{\theta }(D_{\theta} q) -b(\bm{1}^t D_{\theta } q)]   x^+_{\theta}(q) + [v_{\theta }(D_{\theta} q) + M (\bm{1}^t D_{\theta } q)]   x^-_{\theta}(q) \bigg]  }{}
	\qcon{ \qquad \qquad \geq  \sum_q \bigg[ [v_{\theta }(D_{\theta} q) -b(\bm{1}^t D_{\theta } q)]   x^+_{\theta'}(q) + [v_{\theta }(D_{\theta} q) + M (\bm{1}^t D_{\theta } q)]   x^-_{\theta'}(q)  \bigg]  }{ \theta'\neq \theta}{}
	\qcon{  \sum_q [ x^+_{\theta}(q) + x^-_{\theta}(q)] \cdot q = \mu  }{ \theta \in \Theta}{}
	\qcon{  x_{\theta}(q) \ge 0}{\theta, q }{}
\end{lp}

Interestingly, it turns out that solutions to LP \eqref{lp:primal-opt} have a natural interpretation as a \emph{special type} of pricing outcomes mechanism, as formalized in the following lemma.
\begin{lemma}
	Any feasible solution to LP \eqref{lp:primal-opt} is an incentive compatible pricing outcomes mechanism which uses only two different payment methods: either buyer pays $b$ to the seller or the seller pays $M$ to the buyer. 
\end{lemma}
\begin{proof}
	The proof is almost evident from the formulation of LP \eqref{lp:primal-opt}. Any feasible solution to LP \eqref{lp:primal-opt} corresponds to, for any buyer type $\theta$, a convex decomposition of the prior $\mu$ where the posterior is $q$ with probability $x_{\theta}^{+}(q) + x_{\theta}^-(q)$. Moreover, $\frac{x_{\theta}^+(q)}{ x_{\theta}^{+} (q)+ x_{\theta}^-(q) }$ of the time the buyer pays $b$ to the seller in which case the buyer's remaining utility is  $[v_{\theta }(D_{\theta} q) -b(\bm{1}^t D_{\theta } q)] $ (up to a normalization factor). Otherwise, the seller pays $M$ to the buyer in which case the buyer's remaining utility is  $[v_{\theta }(D_{\theta} q)+ M (\bm{1}^t D_{\theta } q)] $. As a result, it is easy to see that the first and second sets of constraints in LP \eqref{lp:primal-opt} are precisely the individual rationality and incentive compatibility constraints. Moreover, it is also easy to see that the objective function is precisely the seller's revenue under this mechanism. This concludes the proof. 
\end{proof} 


Next, we show that LP \eqref{lp:primal-opt} always admits an optimal solution which is a $\texttt{CM-probR}$. More precisely, $x_{\theta}^+, x_{\theta}^-$ support on at most $|A|$ different posteriors, each resulting in a different best response action for buyer type $\theta$. 

\begin{lemma}\label{lem:succint-mechanism}
	There always exists an optimal solution to  LP \eqref{lp:primal-opt} where, for any $\theta \in \Theta$,  both $x_{\theta}^+$ and $x_{\theta}^-$ are supported on at most $|A|$ different posterior $q$'s, each inducing a different best response action for buyer type $\theta$. In other words, LP \eqref{lp:primal-opt} always admits an optimal solution which is a $\texttt{CM-probR}$ as described in Definition \ref{def:BDR} by viewing each posterior distribution as an action recommendation.  
\end{lemma}
\begin{proof}
	We prove the lemma by converting any optimal solution to LP \eqref{lp:primal-opt} to a feasible and optimal solution that satisfies the desired property. 
	
	Let $\mathbf{x}=\{ x_{\theta}^+(q), x_{\theta}^{-}(q) \}_{\theta \in \Theta, q \in Q^*}$ be an arbitrary optimal solution to LP \eqref{lp:primal-opt}. Let $q_1,q_2 \in Q^*$ be two posteriors such that  $x_{\theta}^+(q_1),x_{\theta}^+(q_2) >0$ and buyer utility $v_{\theta}(D_{\theta}q_1),v_{\theta}(D_{\theta}q_2)$ are both achieved by  the same best action $\hat{a}$. We now adjust $\mathbf{x}$ as follows. Denote $$q = \frac{1}{x_{\theta}^+(q_1) + x_{\theta}^+(q_2)} \bigg[ x_{\theta}^+(q_1) \cdot q_1 + x_{\theta}^+(q_2) \cdot q_2 \bigg] \in \Delta_{\Omega}$$ and let $\hat{x}^+_{\theta}(q) = x_{\theta}^+(q_1) + x_{\theta}^+(q_2)$ whereas $\hat{x}_{\theta}^+(q_1) = 0,  \hat{x}_{\theta}^+(q_2) = 0$.  Observe that the buyer $\theta$'s optimal action on posterior $q$ is still $\hat{a}$ because $\hat{a} = \arg \max_{a \in A} \sum_{\omega \in \Omega} (q_1)_{\omega}u(w, \theta, a)$ and $\hat{a} = \arg \max_{a \in A} \sum_{\omega \in \Omega} (q_2)_{\omega}u(w, \theta, a)$  imply that $\hat{a} = \arg \max_{a \in A} \sum_{\omega \in \Omega} [ c_1 (q_1)_{\omega} + c_2  (q_2)_{\omega} ] u(w, \theta, a)$ for any constants $c_1,c_2 >0$. As a result, we have 
	\begin{equation}\label{eq:lin-invariant-v}
	c_1 \cdot v_{\theta}(D_{\theta} q_1) + c_2 \cdot v_{\theta}(D_{\theta} q_2)  = v_{\theta} \big(  D_{\theta} [c_1\cdot q_1+c_2 \cdot q_2]  \big)
	\end{equation}
	
	We now argue that our adjustment to $\mathbf{x}$ maintains feasibility and optimality. The adjustment  does not change the objective because
	\begin{eqnarray*}
		b(\bm{1}^t D_{\theta } q) \cdot \hat{x}^+_{\theta}(q)  &=& b\bigg[\bm{1}^t D_{\theta } \cdot (\hat{x}^+_{\theta}(q) \cdot q) \bigg] \\
		&=& b\bigg[\bm{1}^t D_{\theta } \cdot (x^+_{\theta}(q_1) \cdot q_1 + x^+_{\theta}(q_2) \cdot q_2  ) \bigg]   \\
		& = &    b(\bm{1}^t D_{\theta } q_1) \cdot x^+_{\theta}(q_1) +  b(\bm{1}^t D_{\theta } q_2) \cdot x^+_{\theta}(q_2) 
	\end{eqnarray*} 
	We thus only need to prove that it remains feasible. The adjustment does not change the first constraint, because
	\begin{eqnarray*}
		&& \bigg[v_{\theta }(D_{\theta} q) -b(\bm{1}^t D_{\theta } q) \bigg]  \hat{x}_{\theta}^+(q)  \\
		& = &  u\bigg[ D_{\theta} \cdot [\hat{x}_{\theta}^+(q) \cdot q] \bigg] -  b\bigg[ \bm{1}^t D_{\theta } \cdot (\hat{x}_{\theta}^+(q) \cdot q)    \bigg] \\
		& = & u\bigg[ D_{\theta} \cdot [x_{\theta}^+(q_1) \cdot q_1 + x_{\theta}^+(q_2)\cdot q_2] \bigg] -  b\bigg[ \bm{1}^t D_{\theta } \cdot (x_{\theta}^+(q_1) \cdot q_1 + x_{\theta}^+(q_2) \cdot q_2 )    \bigg] \\ 
		& = & \bigg[v_{\theta }(D_{\theta} q_1) -b(\bm{1}^t D_{\theta } q_1) \bigg]   x_{\theta}^+(q_1)   + \bigg[v_{\theta }(D_{\theta} q_2) -b(\bm{1}^t D_{\theta } q_2) \bigg]   x_{\theta}^+(q_2)
	\end{eqnarray*}
	where the last equality is due to Inequality \eqref{eq:lin-invariant-v}. 
	
	Next,  we show that the second set of constraints also remain feasible under our adjustment. The above argument shows that the value of its left-hand side remains the same for any $\theta, q$. We now show that the adjustment will \emph{not  increase} the value of the right-hand side. The only possibility that the value of the right-hand side could change is when $\theta'$ in $x_{\theta'}^+(q)$ correspond to our focus of buyer type $\theta$ since all other $x_{\theta}^+(q)$'s are not changed. In this case, we have for any $\theta' \not = \theta$ 
	\begin{eqnarray*}
		& & [v_{\theta' }(D_{\theta'} q) -b(\bm{1}^t D_{\theta' } q)]   \hat{x}^+_{\theta}(q)  \\
		& \leq  &  \bigg[v_{\theta' }(D_{\theta'} q_1) \cdot \frac{x_{\theta}^+(q_1)}{x_{\theta}^+(q_1) + x_{\theta}^+(q_2)}  +  v_{\theta' }(D_{\theta'} q_2) \cdot \frac{x_{\theta}^+(q_2)}{x_{\theta}^+(q_1) + x_{\theta}^+(q_2)}  \bigg] \cdot \hat{x}^+_{\theta}(q)    -b(\bm{1}^t D_{\theta' } q) \cdot   \hat{x}^+_{\theta}(q)   \\
		& = & \bigg[v_{\theta' }(D_{\theta'} q_1) \cdot x_{\theta}^+(q_1) +  v_{\theta' }(D_{\theta'} q_2) \cdot x_{\theta}^+(q_2)  \bigg]     -b(\bm{1}^t D_{\theta' } q_1) \cdot   x^+_{\theta}(q_1)  -  b(\bm{1}^t D_{\theta' } q_2) \cdot   x^+_{\theta}(q_2)  \\ 
	\end{eqnarray*}
	where the inequality is due to the convexity of $v_{\theta'}$ function and our definition of $q = \frac{1}{x_{\theta}^+(q_1) + x_{\theta}^+(q_2)} \bigg[ x_{\theta}^+(q_1) \cdot q_1 + x_{\theta}^+(q_2) \cdot q_2 \bigg] $ as a convex combination of $q_1,q_2$. This shows that our adjustment will not increase the value of the right-hand side of the second constraint for any $\theta, \theta'$. It is easy to see that the third set of constraints also hold. As a result, the adjusted solution remains feasible.  
	
	The above adjustment can be done for any $\theta$ and for both $x_{\theta}^+$ and $x_{\theta}^-$. Therefore, ultimately, we can obtain an optimal solution with the desired property as described in the lemma. 
\end{proof}	

The following lemma concludes our proof of the characterization part in Theorem \ref{thm:BDR-opt}.
\begin{lemma}
	The optimal revenue can be achieved by an IC Consulting Mechanism with Probabilistic Return ($\texttt{CM-probR}$).
\end{lemma}
\begin{proof}
	This is because the optimal objective of LP \eqref{lp:primal-opt} --- i.e., the maximium revenue acheived by the optimal $\texttt{CM-probR}$ --- equals the optimal objective of LP \eqref{lp:posterior-opt-budget}, i.e., the maximum revenue achieved by the optimal pricing outcomes mechanims, which is proved to be an optimal mechanim by Lemma \ref{lem:signal-opt}.  
\end{proof}

\section{Omitted Proofs in Section \ref{sec:black-box} }\label{sec:unknown}
\subsection{Proof of Lemma~\ref{lem:unknown}} \label{app:unknown_lem}

We first prove that $\varepsilon$-IC can be guaranteed by the IC constraint in the LP. Consider any $\mu(\omega, \theta, b)$, any $(\theta, b)$. The expected payoff of a type-$(\theta,b)$ buyer when reporting $\theta,b$ is equal to 
\begin{eqnarray*}
U_{\theta,b}(\theta, b) = \E_{\omega_1\sim \mu(\omega|\theta, b)} \ \E_{T\sim\mu^{n-1}} \  \E_{\mathcal{M}_S} \left[u(\omega_1, \theta, a_S(\theta, b)) -  t_S(\theta, b) \right].
\end{eqnarray*}
Consider this expected payoff conditioning on $S_{\theta,b}$. Let $\mathcal{U}(S_{\theta,b})$ be the uniform distribution over $S_{\theta, b}$. Then
\begin{eqnarray*}
U_{\theta,b}(\theta, b | S_{\theta,b}) & = \E_{\omega_1\sim \mu(\omega|\theta, b)} \ \E \left[u(\omega_1, \theta, a_S(\theta, b)) -  t_S(\theta, b) ~\vert~ S_{\theta,b}\right] \\
  & = \E_{\omega_1\sim \mathcal{U}(S_{\theta,b})} \  \E \left[u(\omega_1, \theta, a_S(\theta, b)) -  t_S(\theta, b) ~\vert~ S_{\theta,b}\right]
\end{eqnarray*}
Conditioning on $S_{\theta,b}$, we should have $\omega_1\sim \mathcal{U}(S_{\theta,b})$ because $\omega_1 \in S_{\theta, b}$ and the elements in $S_{\theta,b}$ are independently drawn from $\mu(\omega|\theta,b)$. Recall that $\widehat{\mu}(\omega|\theta,b)$ is defined to be the uniform distribution over $S_{\theta,b}$. So our IC constraints for type-$(\theta, b)$ buyer guarantees that  when $\omega_1$ is drawn from the uniform distribution over $S_{\theta,b}$, it is $\varepsilon$-optimal to truthfully report $(\theta,b)$, i.e.,
$$
U_{\theta,b}(\theta, b | S_{\theta,b}) \ge U_{\theta,b}(\theta', b' | S_{\theta,b}) - \varepsilon
$$
for all $(\theta', b')\neq (\theta, b)$. Here  
$$
U_{\theta,b}(\theta', b' | S_{\theta,b}) = \E_{\omega_1\sim \mathcal{U}(S_{\theta,b})} \  \E \left[u(\omega_1, \theta, a'(a_S(\theta', b'))) -  t_S(\theta', b') ~\vert~ S_{\theta,b}\right]
$$ 
represents the conditional expected utility when reporting $(\theta', b')$ and then taking arbitrary action based on the recommendation.  Take expectation over $S_{\theta,b}$, we prove $\varepsilon$-IC. 

The proofs for $\varepsilon$-IR and $\varepsilon$-obedience are basically the same. $\varepsilon$-IR is guaranteed by the IR constraints for each $S_{\theta, b}$. 
$\varepsilon$-obedience is guaranteed by the 
OB constraints for each $S_{\theta, b}$.

\subsection{Proof of Theorem~\ref{thm:unknown}} \label{app:unknown_thm}
We first show that the  
optimal solution of~\eqref{prog:CMPR} is a feasible solution of~\eqref{prog:unknown} with high probability.

\begin{lemma}
With probability at least $1-\delta/2$, the optimal solution of~\eqref{prog:CMPR} $p^*$  is a feasible solution of the LP~\eqref{prog:unknown}, as long as
$$
n \ge \Theta\left( \ln(G/\delta) \cdot \max\left\{ \frac{|A|^2 }{\varepsilon^2 \cdot \mu_{\min}}, \frac{1}{\mu_{\min}^2} \right\}\right) = \Theta\left( \frac{|A|^2 \cdot \ln(G/\delta)}{\varepsilon^2 }\right)
$$
where $G = \max \{ |\Theta|, |B|, |A|\}$ and $\mu_{\min} = \min_{\theta, b} \mu(\theta, b)$.
\end{lemma}
\begin{proof}
 Define $\mu_{\min} = \min_{\theta, b} \mu(\theta, b)$. We first show that with probability at least $1-\delta/4$, we collect at least $\frac{n \cdot \mu_{\min}}{2}$ samples for each buyer type $\theta, b$. 
By the Chernoff bound, for each $\theta, b$, with probability at least $1 - \delta/(4|\Theta|\cdot |B|)$, 
\begin{eqnarray*}
&   \mu(\theta, b) - \frac{\sum_{i=1}^n \bm{1}\{(\theta_i, b_i) = (\theta, b) \} }{n} \le \sqrt{\frac{\ln(4|\Theta|\cdot |B|/\delta)}{2n}}\\
\Longrightarrow & \frac{\sum_{i=1}^n \bm{1}\{(\theta_i, b_i) = (\theta, b) \} }{n} \ge  \mu(\theta, b) - \sqrt{\frac{\ln(4|\Theta|\cdot |B|/\delta)}{2n}}\\
\Longrightarrow &  \frac{\sum_{i=1}^n \bm{1}\{(\theta_i, b_i) = (\theta, b) \} }{n} \ge  \mu_{\min} - \sqrt{\frac{\ln(4|\Theta|\cdot |B|/\delta)}{2n}}
\end{eqnarray*}
So by the union bound when $n \ge \frac{2\ln(4|\Theta|\cdot |B|/\delta)}{\mu_{\min}^2} $ we have $\frac{\sum_{i=1}^n \bm{1}\{(\theta_i, b_i) = (\theta, b) \} }{n} \ge \mu_{\min}/2$ for all $\theta,b$ with probability $1-\delta/4$.

Now assume that we collect at least $\frac{n \cdot \mu_{\min}}{2}$ samples for each buyer type $\theta, b$. Let $p^{*, \circ}_{\theta,b}(\omega, a)$ be the variables of the optimal solution of~\eqref{prog:CMPR} $p^*$, where $\circ \in \{+,-\}$. To prove that $p^*$ is a feasible solution of~\eqref{prog:unknown} with high probability, it suffices to show that with high probability, we have
$$
\left| \sum_\omega ( \widehat{\mu}(\omega | \theta, b) - \mu(\omega | \theta, b) ) p^{*,\circ}_{\theta, b}(\omega, a) (u( \omega, \theta, a) - t^\circ) \right| \le \frac{1}{4|A|}
$$
for all $\theta, \theta', b, b', a, a'$. Now consider the buyer with type $\theta$ and budget $b$. Let $K$ be the number of samples with $(\theta_i, b_i) = (\theta, b)$. Consider fixed $\theta', b', a, a'$ and define $h^\circ_\omega = p^{*,\circ}_{\theta', b'}(\omega, a) (u( \omega, \theta, a') - t^\circ)$. According to the Chernoff bound, with probability at least $1 - \delta/(8|\Theta|^2 |B|^2  |A|^2)$,
\begin{eqnarray*}
\left|\frac{\sum_{i=1}^n \bm{1}\{ (\theta_i, b_i) = (\theta, b)\} h^\circ_{\omega_i} }{K} - \sum_\omega \mu(\omega | \theta, b) h^\circ_\omega \right| \le \sqrt{\frac{2\ln(8|\Theta|^2 |B|^2 |A|^2/\delta)}{K}} \le \sqrt{\frac{4\ln(8|\Theta|^2 |B|^2 |A|^2/\delta)}{n \cdot \mu_{\min}}}.
\end{eqnarray*}
Setting the right hand side$=\frac{1}{4|A|}$, we get $n \ge  \frac{64 |A|^2 \cdot \ln(8|\Theta|^2 |B|^2 |A|^2/\delta)}{\varepsilon^2 \cdot \mu_{\min}}$. And by the union bound, with probability $1-\delta/4$ we have the left-hand side $\le \frac{1}{4|A|}$ for all $\theta, \theta', b, b', a, a', \circ\in\{+,-\}$. 


Therefore when
$$
n \ge \max\left\{ \frac{64 |A|^2 \cdot \ln(8|\Theta|^2 |B|^2 |A|^2/\delta)}{\varepsilon^2 \cdot \mu_{\min}},  \frac{2\ln(2|\Theta|\cdot |B|/\delta)}{\mu_{\min}^2}\right\} = \Theta\left( \ln(G/\delta) \cdot \max\left\{ \frac{|A|^2 }{\varepsilon^2 \cdot \mu_{\min}}, \frac{1}{\mu_{\min}^2} \right\}\right),
$$
with probability $1-\delta/2$, $p^*$ is a feasible solution of~\eqref{prog:unknown}, where $G = \max \{ |\Theta|, |B|, |A|\}$. 
\end{proof}

Define $\texttt{Rev}_S(\mathcal{M})$ to be the expected revenue of $\mathcal{M}$ when the underlying distribution is the uniform distribution over $S$. Let $\mathcal{E}$ be the event that  $p^*$  is a feasible solution of the LP~\eqref{prog:unknown}. Then the expected revenue of Mechanism~\ref{alg:unknown} is equal to 
\begin{eqnarray*}
&\E_S \ \texttt{Rev}_S(\mathcal{M}_S)  &=  \E_S \ \mathds{1}(\mathcal{E}) \texttt{Rev}_S(\mathcal{M}_S) + \E_S \ \mathds{1}(\overline{\mathcal{E}}) \texttt{Rev}_S(\mathcal{M}_S)\\
&& \ge   \E_S \  \mathds{1}(\mathcal{E})\  \texttt{Rev}_S(p^*) + \E_S \  \mathds{1}(\overline{\mathcal{E}})  \big(\texttt{Rev}_S(p^*)+ \texttt{Rev}_S(\mathcal{M}_S) -  \texttt{Rev}_S(p^*)\big)\\
&& =  \E_S \  \texttt{Rev}_S(p^*) + \E_S \  \mathds{1}(\overline{\mathcal{E}})  \big(\texttt{Rev}_S(\mathcal{M}_S) -  \texttt{Rev}_S(p^*)\big) \\
&& \ge  \E_S \  \texttt{Rev}_S(p^*) - 2 \cdot \E_S \  \mathds{1}(\overline{\mathcal{E}}) \\
&& \ge  \texttt{Rev}_\mu(p^*) - \delta.
\end{eqnarray*}

\end{document}